\documentclass[11pt]{article}
\usepackage{amsthm,amssymb,amsfonts,amsmath,amscd}
\usepackage{nccmath}
\usepackage{graphicx}
\usepackage[pdfa]{hyperref}

\makeatletter \@addtoreset{equation}{section} \makeatother

\newtheorem{theorem}{Theorem}[section]
\newtheorem{lemma}{Lemma}[section]

\newtheorem{proposition}{Proposition}[section]

\newcommand{\E}{\mathbf{E}}
\newcommand{\mdet}{\mathrm{det}}

\newcommand{\Tr}{\mathrm{Tr}\,}

\begin{document}

\title{Characteristic polynomials of non-Hermitian random band matrices near the threshold}
\author{ Mariya Shcherbina
\thanks{School of Mathematics in University of Bristol and Institute for Low Temperature Physics, Kharkiv, Ukraine, e-mail: shcherbi@ilt.kharkov.ua. The
research was supported by the ERC Advanced Grant "RMTBeyond" No. 101020331.  MS completed the work during her stay as Royal Society Wolfson Visiting Fellow (RSWVF/25/R2/1008) at the School of Mathematics in Bristol University.} \and
 Tatyana Shcherbina
\thanks{ Department of Mathematics, University of Wisconsin - Madison, USA, e-mail: tshcherbyna@wisc.edu. This material is based upon work supported  in part by Alfred P. Sloan Foundation grant FG-2022-18916 and the National Science Foundation  grant DMS-2346379}
}
\date{}
\maketitle

\begin{abstract}
The paper \cite{SS:band_com_d} shows that the asymptotic behavior of the second correlation function of characteristic polynomials of the $N\times N$
non-Hermitian random band matrices with a bandwidth $W$ exhibits the transition at $W\sim \sqrt{N}$, as $W,N\to \infty$:  it coincides with those for 
Ginibre ensemble for $W\gg \sqrt{N}$, and
factorized as $1\ll W\ll \sqrt{N}$. 
In this work we extend the techniques of \cite{SS:band_com_d} to study the critical regime when the 
bandwidth $W$ is proportional to the threshold $\sqrt{N}$.

\end{abstract}

\section{Introduction}\label{s:1}
As in \cite{SS:band_com_d}, we consider $N\times N$ non-Hermitian random band matrices (RBM) $H_N$ with independent Gaussian entries $H_{jk}$ such that
\begin{equation}\label{ban}
\E\{H_{jk}\}=0,\quad \mathbf{E}\big\{ H_{jk}\bar H_{jk}\big\}=J_{jk}
\end{equation}
with
\begin{equation}\label{J}
J=\left(-W^2\Delta+1\right)^{-1}.
\end{equation}
Here $\Delta$ is the discrete Laplacian on $[1,N]\cap \mathbb{Z}$ with Neumann boundary conditions.
It is easy to see that  the variance of matrix elements $J_{ij}$ is exponentially small when $|i-j|\gg W$, and so  $W$ can be considered as the  bandwidth.

The probability law of  $H_N$ can be written in the form
\begin{equation}\label{band}
P_N(d H_N)=\prod\limits_{j,k=1}^N\dfrac{dH_{jk}d\overline{H}_{jk}}{\pi J_{jk}}e^{-\frac{|H_{jk}|^2}{J_{jk}}}.
\end{equation}
The Hermitian analog of matrices (\ref{ban}) have various application in mathematical physics: the
eigenvalue statistics of RBM is in relevance in quantum chaos, the quantum dynamics
associated with RBM can be used to model conductance in thick wires, etc. In particular, the local eigenvalue statistics of Hermitian RBM poses 
the transition similar to the famous Anderson metal-insulator phase transition in dimension one: for $W\gg \sqrt{N}$ the eigenvalues have universal GUE 
local statistics and eigenvectors are  delocalized, while $W\ll \sqrt{N}$ gives the Poisson statistics and the localized eigenvectors.
This result was supported by the physical derivation due to Fyodorov and Mirlin  \cite{FM:91}, and rigorously confirmed recently at \cite{SS:Un}, \cite{YY:1d_band}, \cite{ErR:band}, \cite{Dr:band}. The recent higher dimensional results can be found in \cite{YY:2d_band}, \cite{YY:3d_band}.

The non-Hermitian RBM are much less studied. The empirical spectral distribution is expected to converge to the circular law, i.e. to the uniform distribution on a unit disk appearing as a limiting distribution for the non-Hermitian matrices with iid entries (see \cite{TV:08}, \cite{TVKr:10} and references therein), if $W,N\to\infty$.
For RBM with the specific block tridiagonal structure the circular law is proved in \cite{H:25_1} for $W\gg n^\varepsilon$ , however, for the general band profile the best recent result \cite{H:25} shows such convergence for $W\gg N^{1/2+c}$ only (see also \cite{H:24}, \cite{JJLO:21}, \cite{Tikh:23} and references therein for previous results). 

As the first attempt to test the local eigenvalue statistics of non-Hermitian RBM, in the paper \cite{SS:band_com_d} we studied the asymptotic local behaviour of the 
second correlation functions of characteristic polynomials defined as
\begin{equation}\label{Theta_k}
\Theta (z_1,z_2)=\mathbf{E}\Big\{\prod\limits_{s=1}^{2}\det(X_n-z_s)\det(X_n-z_s)^*\Big\},
\end{equation}
where expectation is taken with respect to (\ref{band}), and
\begin{align}\label{z_1,2}
z_1=z+\zeta/N^{1/2},\, z_2=z-\zeta/N^{1/2}, \quad |z|<1,
\end{align}
with $\zeta$ varying in a compact set in $\mathbb{C}$. 

The main result of \cite{SS:band_com_d}  shows the transition in the local behavior of $\Theta(z_1,z_2)$ at $W\sim \sqrt{N}$ similar 
to that for the Hermitian (or real symmetric) RBM (see \cite{TSh:14}, \cite{TSh:15}, \cite{SS:17},  \cite{TS:22})): 
\begin{theorem}\label{t:trans} Given the band matrix of the form (\ref{band}) -- (\ref{J}), and $z_1, z_2$ of (\ref{z_1,2}), we have
with any fixed $\varepsilon_0>0$
\begin{align*}
\lim_{N\to\infty,\frac{W^2}{N\log^2 N}\to \infty}\frac{\Theta(z_1,z_2)}{\Theta^{1/2}(z_1,z_1)\Theta^{1/2}(z_2,z_2)}=\left\{\begin{array}{ll}
\frac{1-e^{-4|\zeta|^2}}{4|\zeta|^2},& \sqrt {N\log^2 N} \le W\le N^{1-\varepsilon_0},\\
 e^{-2|\zeta|^2},& N^{\varepsilon_0}\le W\le \sqrt{\frac{N}{\log N}}.
\end{array}\right.
\end{align*}
The limit in the ``delocalized" regime $\sqrt {N\log^2 N} \le W\le N^{1-\varepsilon_0}$ coincides with the same limit taken for Ginibre matrices with iid Gaussian entries computed in \cite{Ak-Ve:03}. 
\end{theorem}
Notice that although  $\Theta$ is not a local object in terms of eigenvalue statistics, the limit  in the ``delocalized" regime holds also 
for any non-Hermitian matrices with iid entries as soon as the first four moments of elements distribution are Gaussian (see \cite{Af:19}), as well as for sparse
non-Hermitian random matrices (see \cite{AS:sp_det}). This is the manifistation of universality similar to the bulk universality of the local spectral statistics represented by the classical correlation function (see \cite{Ta-Vu:15}, \cite{MO:24}). Similar results were obtained for many classical Hermitian random matrix ensembles
(see, e.g., \cite{Br-Hi:00}, \cite{Br-Hi:01}, \cite{St-Fy:03},\cite{TSh:ChW}, \cite{TSh:ChSC},\cite{Af:16}, etc.)


The aim of the current paper is to complement Theorem \ref{t:trans} by the study of the asymptotic local behaviour of the 
second correlation functions of characteristic polynomials (\ref{Theta_k}) for non-Hermitian RBM (\ref{band}) -- (\ref{J}) near  
the threshold $W\sim \sqrt{N}$. The main result is
\begin{theorem}\label{t:new}
Assume that $W=\kappa_*N^{1/2}(1+o(1))$ as $N\to \infty$. Then
\begin{align}\label{t1.1}
\lim_{N\to\infty}\frac{\Theta(z_1,z_2)}{\Theta^{1/2}(z_1,z_1)\Theta^{1/2}(z_2,z_2)}=(e^{\mathbb{A}_0}\mathbf{1},\mathbf{1})
\end{align}
where the differential operator $\mathbb{A}_0$ has the form
\begin{align}\label{bbA}
\mathbb{A}_0\varphi(z)=\frac{1}{8(\kappa_*u_*)^2}\Big((1-z^2)\frac{d^2}{dz^2}-2z\frac{d}{dz}\Big)\varphi(z)+2z\varphi(z),\quad z\in[-1,1],
\end{align}
with a boundary conditions $|\varphi(1)|<\infty$, $|\varphi(-1)|<\infty$, and $\mathbf{1}(z)\equiv 1$.
\end{theorem}
The proof of Theorem \ref{t:new} is based on the extension of SUSY transfer matrix approach developed in \cite{SS:band_com_d}.
Notice that similar result for the Hermitian/real symmetric RBM was obtained in \cite{TS:20}, \cite{TS:22}.

%
The paper is organized as follows. In Section 2 we  summarize and refine the results of \cite{SS:band_com_d} needed for the proof of Theorem \ref{t:new}.
The proof of Theorem \ref{t:new} is presented in Section 3.

%


We denote by $C$, $C_1$, etc. various $W$ and $N$-independent quantities below, which
can be different in different formulas. To reduce the number of notations, we also use the same letters for the integral operators and their kernels.

\section{Transfer matrix representation and preliminary operator analysis}\label{s:2}
In this section we summarize the results of \cite{SS:band_com_d} needed for the proof of Theorem \ref{t:new}.
For the reader convenience, we provide here the main steps of the transfer operator analysis as well as the sketch of some proofs in order to explain the main ideas without
going deeply into technical details of  \cite{SS:band_com_d}. The reader interested in the full proofs should be referred to \cite{SS:band_com_d}.

Following \cite{SS:band_com_d}, it is more convenient to consider the normalized version of $\Theta$:
\begin{align}\label{ti-theta}
\tilde\Theta(z_1,z_2)=(\pi^2W^2\lambda_*^{-1})^{2(N-1)} C_{N,W}^{-1}\,\Theta(z_1,z_2),
\end{align}
where  
\begin{align}\label{lambda_*}
 \lambda_*=1-W^{-1}(\alpha-u_*^2W^{-1}), \quad \alpha=u_*(2+u_*^2W^{-2})^{1/2},\quad u_*=(1-|z|^2)^{1/2},
\end{align}
and 
\begin{equation}\label{C}
C_{N,W}=\pi^{-4N} \mdet^{-4} J \cdot e^{-2Nu_*^2}.
\end{equation}
We are going to use the following operator form of the integral representation of  $\tilde\Theta(z_1,z_2)$ obtained in \cite{SS:band_com_d} via supersymmetry techniques:
\begin{proposition}\label{prop:IR}
	Let  $H$ be the non-Hermitian Gaussian random band matrices defined by (\ref{ban}) -- (\ref{band}). Then 
	the normalized second correlation function of the characteristic polynomials  $\tilde\Theta$ defined by (\ref{ti-theta}) can be represented in the following form
\begin{align}\label{repr_Theta}
\tilde\Theta(z_1,z_2)=&\int\limits_{(H_2)^N} e^{f(Q_1)}(\prod_{j=1}^{N-1}\mathcal{K}_\zeta(Q_j,Q_{j+1}))e^{ f(Q_N)}
\prod\limits_{j=1}^n dQ_j=(\mathcal{K}_\zeta^{N-1}g,g),
\end{align}
where $H_2$ is the space of $2\times 2$ complex matrices, $\mathcal{H}=L_2(H_2)$, and $\mathcal{K}_\zeta: \mathcal{H}\to\mathcal{H}$ is an integral operator  with the kernel
\begin{align}
\mathcal{K}_\zeta(Q_j,Q_{j+1})=\pi^4W^4\lambda_*^{-2}&\exp\Big\{ -W^2 \Tr (Q_j- Q_{j+1})(Q_j- Q_{j+1})^*+f(Q_j)+f(Q_{j+1})\Big\},\label{K}
\end{align}
where
\begin{align}
 f(Q_j)=&\frac{1}{2}(-\Tr Q_jQ_j^*+\log \det \mathcal{Q}_j +2u_*^2),\quad g(Q)=e^{f(Q)},\label{f}\\
\mathcal{Q}_j =&\left(\begin{array}{cc}\hat z&iQ_j\\iQ_j^*&\hat z^*\end{array}\right),
\quad\hat z=\mathrm{diag}\{z_1,z_2\},
\label{Q_cal}\end{align}
and $\lambda_*$, $u_*$ are defined in (\ref{lambda_*}). 
\end{proposition}

\subsection{Concentration of eigenfunctions of $\mathcal{K}_\zeta$}\label{s:3}
Following \cite{SS:band_com_d}, we notice first that using analyticity of  $\tilde\Theta(z+\zeta,z-\zeta)$ with respect to $\zeta$ and $\bar\zeta$ taken as independent variables and uniqueness theorem, one can consider only $\zeta=\xi e^{i\phi},\bar \zeta=\eta e^{-i\phi}$ with $\phi=\mathrm{arg}\, z$ and  $\xi,\eta\in \mathbb{R}$.
Since by (\ref{f}) $\det \mathcal{Q}_j\in\mathbb{R}$ for such $\zeta, \bar\zeta$, starting from this moment, we can consider $\mathcal{K}_\zeta$ of (\ref{K}) as a positive operator while for simplicity keeping notations $\mathcal{K}_\zeta$, $\zeta,\bar\zeta$.

Next, diagonalizing $Q$ by unitary matrices, it is easy to see that for $\zeta=\bar{\zeta}=0$ the function $f$ of (\ref{f}) takes its maximum at $Q=u_*U$, with some unitary $U$ and $u_*$ of (\ref{lambda_*}).

The first step in the analysis of $\mathcal{K}_\zeta$ is to show that the main contribution to (\ref{repr_Theta}) is given  by  $\cap_j \{\Omega_W(Q_j)\}$ with 
\begin{align}\label{Omega}
\Omega_W=\{Q:\|Q^*Q-u_*^2I_2\|\le \log W/W^{1/2}\}.
\end{align}

\medskip

 Let $\mathbb{P}_W=\mathbf{1}_{\Omega_W}$ be the orthogonal projection in $\mathcal{H}=L_2(\mathbb{C}^4)$ on 
functions whose support lies in the domain $\Omega_W$ of (\ref{Omega}). Then we have
\begin{lemma}\label{l:conc} There is $N,W$-independent $C_1$ such that
\begin{align}\label{Om1}
\|(1-\mathbb{P}_W)\mathcal{K}_\zeta(1-\mathbb{P}_W)\|\le 1-C_1 \log W/W.
\end{align}
\end{lemma}

Now let us study $\mathbb{P}_W\mathcal{K}_\zeta\mathbb{P}_W$. 

Consider the cylinder  change of variables (see, e.g., \cite{Hu:63}) 
\begin{align}\label{c_ch}
Q_i=U_i\mathcal{R}_i, \quad U_i\in U(2), \quad \mathcal{R}_i>0,\quad J(\mathcal{R})=\pi^3(\Tr \mathcal{R})^2\det \mathcal{R}.
\end{align}
Everywhere below we consider our operators acting in $\mathcal{H}_0\otimes L_2(U(2))$ with
\begin{align}\label{H_0}
\mathcal{H}_0=L_2(\mathcal{H}_{2,+}) \,\text{with innner product}\, (\psi_1(\mathcal{R}),\psi_2(\mathcal{R}))=
\int_{\mathcal{H}_{2,+}}\psi_1(\mathcal{R})\overline{\psi_2(\mathcal{R})})d\mathcal{R}.
\end{align}
Here $\mathcal{H}_{2,+}$ is the space of all positive $2\times 2$ matrices and $d\mathcal{R}$ means the Lebesgue measure on $\mathcal{H}_{2,+}$.

Since $\mathbb{P}_W$ is the projector on $\Omega_W$ (see (\ref{Omega})),  Lemma \ref{l:conc} implies that we can restrict the integration with respect to $\mathcal{R}$ by $O(W^{-1/2}\log W)$ neighbourhood
 of $u_*I_2$, i.e. 
 \begin{align}\label{R}
 \mathcal{R}_i=u_*(I_2+W^{-1/2}R_i),\quad R_i=R_i^*,\quad\|R_i\|\le\log W+o(1).
 \end{align}
 Then we get
\begin{align}\label{Theta1}
\tilde\Theta(z_1,z_2)=&(\mathcal{K}_\zeta^{N-1}g,g), 
\end{align}
where $\mathcal{K}_\zeta$ is an integral operator with the kernel
\begin{align}\label{cal_K}
&\mathcal{K}_\zeta(R_1,U_1,R_2,U_2)= \mathcal{A}_\zeta(R_1,U_1,R_2,U_2) K_{R_1,R_2}(U_2^*U_1),
\\
\label{K,Z}
&K_{R_1,R_2}(U)=Z^{-1}(R_1,R_{2})e^{k_{R_1,R_2}(U^*_{2}U_1)},\\
\notag &k_{R_1,R_2}(U)=u_*^2W^2\Tr \Big((U-1)(1+R_1/W^{1/2})(1+R_{2}/W^{1/2})\Big)+cc,\\
& Z(R_1,R_{2})=(\pi u_*W)^{3}\int dU \exp\{k_{R_1,R_2}(U)\},\notag\\
  &\mathcal{Z}(R_1,R_2)=J^{1/2}(u_*(1+R_1/W^{1/2}))J^{1/2}(u_*(1+R_2/W^{1/2}))Z(R_1,R_{2}),
\label{cal_Z}\end{align}
with $J(\mathcal{R})$ defined in (\ref{c_ch}).
Operator $\mathcal{A}_\zeta$ of (\ref{cal_K}) has the form
\begin{align}\label{A}
\mathcal{A}_\zeta(R_1,U_1,R_2,U_2)=&e^{f_\zeta(R_1,U_1)}B(R_1-R_2)
e^{f_\zeta(R_2,U_2)}\mathcal{Z}(R_1,R_2)\\
f_\zeta(R,U)=&f(u_*U(1+R/W^{1/2})),\quad B(R)=(\lambda_*\pi u_*^2W)^{-2}e^{-Wu^2_*\Tr R^2},\notag
\end{align}
where $u_*,\lambda_*$ are defined in (\ref{lambda_*}), and $f$ is from (\ref{f}).

Function $g$ in (\ref{Theta1}) is obtained by the change of variables  (\ref{c_ch}) and (\ref{R}) in  $g$ of (\ref{f}):
\begin{align}\label{new_g}
g=e^{f_\zeta(R,U)},\quad \|g\|=CW(1+o(1)).
\end{align}
Now let us expand $f_\zeta(R,U)$ around  $Q_*=u_*U$.  Denote 
\begin{align}
L_{U^*}=U^*LU, \quad \epsilon=\big(W/N\big)^{1/2},\quad\mathcal{M}(U)=-\frac{1}{2u_*^2}(\zeta \bar zL_{U^*}+\bar\zeta  zL).
\label{M(U)}\end{align}
Notice that in the conditions of Theorem \ref{t:new} we have $\varepsilon\sim W^{-1/2}$.

%
Then,
\begin{align}\label{pert}
f_\zeta(R,U)
=&-\frac{u_*^2}{2W}\Tr (R-\epsilon\mathcal{M}(U))^2+N^{-1}\nu(U)+\tilde f_\zeta (R,U)+O(\epsilon^2 W^{-3/2}),
\end{align}
where
\begin{align}\label{ti-f}
\tilde f_\zeta (R,U)=&W^{-3/2}\Tr R^3\varphi_0(1+R/W^{1/2})+(\epsilon/ W^{3/2})\Tr \mathcal{M}(U) R^2\varphi_1(1+R/W^{1/2}),\\
\nu(U)=&|\zeta|^2\Tr LU^*LU/2.
\label{nu}\end{align}
Here $\varphi_0(x)$ and $\varphi_1(x)$ are some $N,W$-independent function analytic around $x=1$, whose concrete form is not important for us.
We also denote  $\hat\nu(U)$ the operator of multiplication  by $\nu$.

Operator $\mathcal{A}_\zeta$ of  (\ref{A})  takes the form
\begin{align}\label{A_zeta}
&\mathcal{A}_\zeta(R_1, U_1,R_2,U_2)=F_\zeta(R_1, U_1)B(R_1-R_2)\mathcal{Z}(R_1,R_2)F_\zeta(R_2, U_2)\big(1+O(N^{-1}W^{-1/2})\big)
\\
&F_\zeta(R, U)= e^{-2u_*^4\Tr (R-\epsilon\mathcal{M}(U))^2/W+\nu(U)/N+\tilde f(R,U)},\qquad\quad F_0(R)=F_\zeta(R, U)\Big|_{\zeta=\bar\zeta=0}.
\notag
\end{align}

\subsection{Analysis of $\mathcal{A}_\zeta$ of (\ref{A_zeta}) and $K_{R_1,R_2}$ of (\ref{K,Z})}\label{s:4}
Consider
\begin{align}\label{hat_A}
\mathcal{A}(R_1,R_2)=&F_0(R_1)B(R_1-R_2)\mathcal{Z}(R_1,R_2)F_0(R_2)=\mathcal{A}_\zeta(R_1,R_2)\Big|_{\zeta=\bar\zeta=0},\\
\mathcal{A}_0(R_1,R_2)=&F_0(R_1)B(R_1-R_2)F_0(R_2)\label{A_0}\\
=&e^{W^{-3/2}\Tr R_1^3\varphi_0(1+R_1/W^{1/2})}\mathcal{A}_*(R_1,R_2)e^{W^{-3/2}\Tr R_2^3\varphi_0(1+R_2/W^{1/2})}\notag\end{align}
 with $\varphi_0$ of (\ref{ti-f}) and $\mathcal{A}_*$ being the quadratic approximation of $\mathcal{A}$:
\begin{align}
&\mathcal{A}_*(R_1,R_2)=e^{-u_*^4\Tr R_1^2/W}B(R_1-R_2)e^{-u_*^4\Tr R_2^2/W}\label{A_*},\\\
&\hskip1.8cm=\mathcal{A}_{*1}(x_{01},x_{02})\mathcal{A}_{*1}(x_{11},x_{12})\mathcal{A}_{*1}(x_{21},x_{22})\mathcal{A}_{*1}(x_{31},x_{32}),\notag\\
&\mathcal{A}_{*1}(x,y)=\Big(\frac{u_*^2W}{\pi\lambda_* }\Big)^{1/2}e^{-2u_*^4x^2/W}e^{-2Wu^2_*(x-y)^2}e^{-2u_*^4y^2/W}.
\notag\end{align}
In (\ref{A_*})  we  represent 
 \begin{align}\label{expR}
 R_{l}=x_{0l}I_2+x_{1l}\sigma_1+x_{2l}\sigma_2+x_{3l}\sigma_3,\quad l=1,2,
 \end{align}
where $\sigma_1,\sigma_2,\sigma_3$  are the Pauli matrices.

Eigenvalues and eigenvectors of $\mathcal{A}_*$
of (\ref{A_*})  can be computed explicitly via Hermite polynomials, and the next lemma shows that their deviation from  eigenvalues and eigenvectors of operator $\mathcal{A}$ is small:
\begin{lemma}\label{l:A}
Let $\{\Psi_{*\bar m}(R),\lambda_{*\bar m}\}$  be eigenvectors and eigenvalues
of the operator $\mathcal{A}_*$ of (\ref{A_*}). Then
\begin{align}
 \Psi_{*\bar m}(R)=&P_{\bar m}(R)e^{-\alpha u_*^2\Tr R^2},\quad P_{\bar m}(R)=\prod_{i=0}^3H_{m_i}(u_*(2\alpha)^{1/2} x_i)/ \kappa_{m_i}, \label{psi_k}\\
 \lambda_{*\bar m}=&\lambda_{*}^{|m|},\quad
\bar m=(m_0,m_1,m_2,m_3),\, m_i=0,1,\dots,\quad |\bar m|=\sum_{i=0}^3|m_i|.
\notag\end{align}
Here $H_m(x)$ is the $m$th Hermite polynomial, $\kappa_m$ is a normalization factor, and $\lambda_*$, $\alpha$ are defined in (\ref{lambda_*}).

Let $E_{|\bar m|}=\mathrm{Lin}\{\Psi_{*\bar j}\}_{|\bar j|=|\bar m|}$ and $\gamma(|\bar m|)=\mathrm{dim}E_{|\bar m|}$. Then  there are   $\gamma(|\bar m|)$
 eigenvalues $\{\lambda^{(\mu)}_{|\bar m|}\}_{\mu=1}^{\gamma(|\bar m|)}$ of $\mathcal{A}$ of (\ref{hat_A}) 
 such that
\begin{align}
&|\lambda^{(\mu)}_{|\bar m|}-\lambda_{*\bar m}|\le C(|\bar m|+1)W^{-2}. \label{diff_eig}\end{align}
and if 
  an eigenvector $\Psi_{\bar 0}(R)$ corresponds to  $\lambda_{max}$, then 
\begin{align}
\Psi_{\bar 0}(R)=\Psi_{*\bar 0}(R)+W^{-1/2}\tilde \Psi_{\bar 0}(R),\quad \|(1-P_L)\Psi_{\bar 0}\|\le CN^{-3}
\label{de_Psi.0}
\end{align}
where we denote by $P_L$  the orthogonal projection in $\mathcal{H}_0=L_2(\mathcal{H}_{2,+})$  on the subspace \\
$\mathcal{H}_L=\mathrm{Lin}\{\Psi_{*\bar k}(R)\}_{|\bar k|\le L}$, $ L=C\log^2W$.
\end{lemma}

Next we show that $\mathcal{Z}(R_1,R_2)$ in the definition (\ref{A}) of $\mathcal A$ of  can be changed by $1$ with a small correction.
We also compare $\mathcal{Z}(R_1,R_2)$ with ``shifted" $\mathcal{Z}(R_1-\epsilon\mathcal{M},R_2-\epsilon\mathcal{M})$.
\begin{lemma}\label{l:Z_0} 

Given $\mathcal{Z}(R_1,R_2)$ of the form (\ref{cal_Z}) and
 $\Psi(R)\in \mathrm{Lin}\{\Psi_{*\bar k}\}_{|\bar k|\le m}$  ($m\ge 0$), we have
\begin{align}\label{act_D}
&\int \mathcal{A}_0(R_1,R_2)\Big(\mathcal{Z}(R_1,R_2)-1\Big) \Psi(R_2)dR_2=O((m+1)W^{-2}\|\Psi\|),
\end{align}
where   $\mathcal{A}_0(R_1,R_2)$  was defined in  (\ref{A_0}).

In addition, for every fixed $2\times 2$ matrix $\mathcal{M}=\mathcal{M}^*$ and $\epsilon=(W/N)^{1/2}$
\begin{align}\label{diff_Z}
\int\mathcal{A}_0(R_1,R_2)\Big(\mathcal{Z}(R_1,R_2) -\mathcal{Z}(R_1-\epsilon\mathcal{M},R_2-\epsilon\mathcal{M})\Big)\Psi(R_2)dR_2
=O(\epsilon W^{-2}\|\Psi\|).
\end{align}
 
\end{lemma}
Now we are going to study the ``unitary part" $K_{R_1,R_2}$ of  operator $\mathcal{K}_\zeta$ (see (\ref{K,Z})), in particular, we need to know to what extend the operator
depend on $R_1,R_2$, if $|R_1-R_2|\sim W^{-1/2}$ (this condition is guaranteed by the factor $e^{-Wu_*^2\Tr (R_1-R_2)^2}$ in (\ref{A})).

Since the kernel of $K_{R_1,R_2}(U_1^*U_2)$  depends on
the "matrix difference" $U_1^*U_2$ only, we conclude that for any $\ell$ the operator  is reduced by the $\ell$th  space of the irreducible representation of $SU(2)$  
which is a linear span of  the coefficients $\{t_{pk}^{(\ell)}(U)\}_{p,k=-\ell}^\ell$ of
 the $\ell$th irreducible representation $T^{(\ell)}(U)$ of $SU(2)$. Hence, eigenvectors of $K_{R_1,R_2}(U_1^*U_2)$
are  linear combinations  of $\{t_{pk}^{(\ell)}(U)\}_{p,k=-\ell}^\ell$, and we study the action of $K_{R_1,R_2}$ on these functions.

Consider some $M\gg 1$ and set
\begin{align}\label{E^ell}
\mathcal{E}^{(\ell)}=\hbox{Lin}\{t^{(\ell)}_{0p}(U)\}_{p=-\ell}^\ell,\quad
\mathcal{E}_{M}=
\cup_{\ell=0}^{M}\mathcal{E}^{(\ell)},\quad \mathcal{E}=\cup_{\ell} \mathcal{E}^{(\ell)}.
\end{align}
We will denote also by $\mathcal{E}_{M}$ the orthogonal projection on $\mathcal{E}_{M}$, and by $\mathcal{E}^{(\ell)}$ the orthogonal projection
on $\mathcal{E}^{(\ell)}$. Notice also that since the function $g$ of  (\ref{f})  and $\mathcal{K}_\zeta$ depend on $U$ only  via $L_{U^*}$,  $\mathcal{K}_\zeta$  can be considered as an operator acting in $\mathcal H_0\otimes\mathcal{E}$ (recall that $\mathcal{H}_0$
 means the $L^2$-space on all positive $2\times 2$ matrices). 
 \begin{lemma}\label{l:KPsi}
Given operator $K_{R_1,R_2}$ with a kernel (\ref{K,Z}), we have for $\|R_1-R_2\|\le W^{-1/2}\log W$ and $\ell\le CW^{3/4}\log W$ with any fixed $C$:
\begin{align}\label{Kt.0}
K_{R_1,R_2}t^{(\ell)}_{0k}=&\tilde\lambda_{\ell}t^{(\ell)}_{0k}
+b^{(\ell)}_{k+1}t^{(\ell)}_{0k+1}+b^{(\ell)}_{k}t^{(\ell)}_{0k-1}
+O(\ell^2\log^2W/W^3),\\
\tilde\lambda_{\ell}=&\lambda_\ell+O(\ell^2/W^2)\big(O(R_1/W^{1/2})+O(R_2/W^{1/2})\big),\notag\\
b^{(\ell)}_{k}=&d^{(\ell)}_{k}(\ell/W) [R_2-R_1,R_1]_{12}+O(\ell W^{-5/2}\log^3 W),
\notag
\end{align}
where $d^{(\ell)}_{k}$ are some bounded constants which are not important for us, and
\begin{equation}\label{lam_l}
\lambda_\ell=1-\ell(\ell+1)/8(u_*W)^{2}.
\end{equation}
Moreover, for any function   $\Psi_h(R,U)=\Psi(R)h(U)$
with  $\Psi(R)\in \mathcal {H}_L$ (see (\ref{P_L})) and  \\
  $h\in\mathcal{E}^{(\ell)},$ we have
\begin{align}
\Big\|(\mathcal{K}_0\Psi_h)(U_1,R_1)-\lambda_{\ell}h(U_1)( \mathcal{A}\Psi)(R_1)\Big\|
\le C\|\Psi_h\|(W^{-1/2}\log W (\ell/W)^2+L \ell/W^2),
\label{KPsi.1}\end{align}
where $\mathcal{K}_0=\mathcal{K}_\zeta\Big|_{\zeta=0}$ and $\mathcal{A}$ was defined in (\ref{hat_A})
 If $\Psi_{0,h}(R,U)=\Psi_{ 0}(R)h(U)$, where $\Psi_0$ is an eigenvector of $\mathcal{A}$ corresponding to $\lambda_{\max}$
  and $h\in\mathcal{E}^{(\ell)},\,\ell\le cW^{3/4},\|h\|=1$, then
\begin{align}
\Big\|(\mathcal{K}_0\Psi_{0,h})-\lambda_\ell\lambda_{\max} \Psi_{0,h}\Big\|
= O(W^{-1/2}(\ell/W)^2).
\label{KPsi0}\end{align}

In addition, for every fixed $2\times 2$ matrix $\mathcal{M}=\mathcal{M}^*$, $\epsilon=(W/N)^{1/2}$, 
$h\in\mathcal{E}_{M},\,M\le W^{1/2}/L$, and $\Psi \in \mathcal{H}_L$
\begin{align}\label{diff_K_eps}
\int\mathcal{A}(R_1,R_2) \Big((K_{R_1,R_2}h)(U_1)- (K_{R_1-\epsilon\mathcal{M},R_2-\epsilon\mathcal{M}}h)(U_1)\Big)&\Psi(R_2)dR_2
=O(\epsilon W^{-2}ML\|\Psi\|),
\end{align}
and 
\begin{align}
\label{b_K_R}
\mathcal{E}^{(\ell)}K_{R_1,R_2}\le
1-CW^{-1/2}\log^2W, \quad \ell>W^{3/4}\log W
\end{align}

\end{lemma}
\textit{Sketch of the proof of Lemmas \ref{l:Z_0} and \ref{l:KPsi}}

We transform $k_{R_1,R_2}(U)$ of (\ref{K,Z}) as
\begin{align*}\notag
k_{R_1,R_2}(U)=
k_{*R_1,R_2}(U)+\rho_1+\rho_2,\notag
\end{align*}
where \begin{align}\label{k_*}
k_{*R_1,R_2}(U)=&(u_*W)^2\Tr S\,\Tr((U+U^*)/2-1)),\quad S=\frac{1}{2}\{1+R_1/W^{1/2},1+R_2/W^{1/2}\}\\
\label{rho}
\rho_1=&2u_*^2W^2\Tr \frac{(U+U^*)^\circ}{2} S^\circ,
\qquad\qquad\rho_2=u_*^2W\Tr \frac{(U-U^*)^\circ}{2}[R_1 , R_2].
_*\end{align}
If $U$ is parametrised as
 \begin{align}\label{U}
 U=&\mathcal{T}(\phi)\left(
\begin{array}{cc}\cos(\theta/2)&i\sin(\theta/2)\\
i\sin(\theta/2)&\cos(\theta/2)\end{array}\right)\mathcal{T}(\psi)e^{i\gamma}
\end{align}
where  $\mathcal{T}(\phi)=\mathrm{diag}\{e^{i\phi/2},e^{-i\phi/2}\}$, $\gamma\in [-\pi/2,\pi/2]$, $\sigma,\delta\in[-\pi,\pi]$, $\theta\in[0,\pi]$, 
then  
\begin{align}\label{k_*.1}
k_{*R_1,R_2}(U)=-2(u_*W)^2\Tr S(1-\cos\gamma\cos(\theta/2)\cos\sigma),\quad \sigma=\frac{1}{2}(\varphi+\psi),\, \delta =\frac{1}{2}(\varphi-\psi),
\end{align}

Hence, the essential contribution to the integrals comes only from the domain where
\[
\sin\gamma=O(W^{-1})\quad \sin\sigma=O(W^{-1})\quad \sin(\theta/2)=O(W^{-1})\quad 
\]
Moreover,

\begin{align}\notag
&\frac{1}{2}(U+U^*)^\circ =-\sin\gamma\,\tilde U,\quad \frac{1}{2}( U- U^*)^\circ=
i\cos\gamma\, \tilde U,\\
&\tilde U=\left(
\begin{array}{cc}\cos(\theta /2)\sin\sigma&e^{i\delta }\sin(\theta/2) \\e^{-i\delta }\sin(\theta/2)
&-\cos(\theta/2)\sin\sigma\end{array}\right).\, 
\notag
\end{align}
Therefore, using these relations and taking into account that $S^\circ\sim W^{-1/2}$ and\\ $[R_1,R_2]=[R_1-R_2,R_2]\sim W^{-1/2}$, we obtain
\begin{align}
&|\rho_1|\le W^{-1/2}\log W,\quad |\rho_2|\le W^{-1/2}\log W,
\label{b_rho}\end{align}
and, thus, we can expand $\exp\{k_{R_1,R_2}(U)\}$ with respect to $\rho_1,\rho_2$. The relations (\ref{act_D}) can be obtained by the
expansion with respect to $\rho_1,\rho_2$ up to the order $O(W^{-2})$ and integrating them with respect to $U$ (which means with respect to
 $\theta,\gamma,\sigma,\delta$). Then  one need to take into account  that the term which appears from $\rho_2^2$ after integration with respect to $U$
 can be integrated by parts with respect to $R_2$. Since $\mathcal{A}(R_1,R_2)$ has the form (\ref{A}), we have
\begin{align}\label{int_com}
&\int \mathcal{A}(R_1,R_2)\Tr[R_2-R_1,R_1][R_1,R_2-R_1]\Psi_{\bar k}(R_2)dR_2\\
&\quad =\frac{2\Tr (R_1^\circ)^2}{u_*^2W}\Psi_{\bar k}(R_1)+O(W^{-2}\|\Psi_{\bar k}''\|)+O(W^{-5/2}).
\notag\end{align}
For the detailed computations see Lemma 4.1 of \cite{SS:band_com_d}.

 Now we are going to study the action of $K_{R_1,R_2}$.
 If $U$ is parametrizes as (\ref{U}),
then the coefficients $\{t_{pk}^{(\ell)}(U)\}_{p,k=-\ell}^\ell$ of
 the $\ell$th irreducible representation $T^{(\ell)}(U)$ of $SU(2)$ can be represented as (see \cite{Vil:68})
\begin{align}\label{as_P}
&t_{sk}^{(\ell)}(\tilde U)=P_{sk}^{(\ell)}(\cos\theta) e^{i(s\varphi+k\psi)},
\end{align}
where $P_{sk}^{(\ell)}(\cos\theta)$ is the associated Legendre polynomial.

To prove Lemma \ref{l:KPsi}, we notice first that
an application of $K_{R_1,R_2}$ of (\ref{K,Z})  to $t_{0k}^{(\ell)}$ and changing the integration variable $U_2\to U_1U^*$, yields
\begin{align}\label{Kt}
&({K}_{R_1,R_2}t_{0k}^{(\ell)},t_{0k'}^{(\ell)})=
Z^{-1}\int\exp\{k_{R_1,R_2}(U_2^*U_1)\}t_{0p}^{(\ell)}(U_2)dU_2
=\sum_{s}t_{0s}^{(\ell)}(U_1)\mathcal F_{sk}^{(\ell)}({R_1,R_2}),
\end{align}
 where
\begin{align}\notag
&\mathcal{F}_{sk}^{(\ell)}(R_1,R_2)=Z^{-1}\int \exp\{k_{R_1,R_2}( U)\}t_{sk}^{(\ell)}( U^*)dU.
\end{align}
Then we analyse
\begin{align}\label{F_ks.1}
\mathcal{F}_{sk}^{(\ell)}(R_1,R_2)=&\frac{\left\langle t_{sk}^{(\ell)}(\tilde U^*)(1+\sum (\rho_1+\rho_2)^m/m!)\right\rangle}
{\left\langle 1+\sum  (\rho_1+\rho_2)^m/m!\right\rangle},
\end{align}
where $\rho_1, \rho_2$ are defined in (\ref{rho}), and  for any $f(U)$ we denote
\begin{align}\label{<f>}
\left\langle f\right\rangle=&Z_0^{-1}\int dU f(U)\exp\{-2u_*^2W^2\Tr S(1-\cos(\theta/2)\cos\sigma\cos\gamma)\},\\
Z_0=&\int dU \exp\{-2u_*^2W^2\Tr S(1-\cos(\theta/2)\cos\sigma\cos\gamma)\}\notag\\=&(2\pi u_*^2W^2\Tr S )^{-2}(1+O(W^{-2})).
\label{Z_0}\end{align}

To analyse $\mathcal{F}_{sk}^{(\ell)}(R_1,R_2)$, we use the representation (\ref{as_P})
and the following  asymptotic relations (see \cite{Vil:68})
\begin{align}\label{as_Leg.0} 
P_{k+1,k}^{(\ell)}(\cos\theta)=&i(1+(k+1)/\ell)^{1/2}(1-k/\ell)^{1/2}\ell \sin(\theta/2)(1+O(\ell \sin^2(\theta/2))\\
|P_{k+q,k}^{(\ell)}(\cos\theta)|\le& (\kappa\ell\sin(\theta/2))^q,\quad q\ge 2\notag\\
P^{(\ell)}_{00}(\cos\theta)=&1-\ell(\ell+1)\sin^2(\theta/2)+O((\ell\sin(\theta/2))^3),
\label{as_Leg.1}\end{align}
where $\kappa>0$ is some fixed constant. 

In addition, for any $\ell>W^{3/4}\log^2W$ , if $\|R_1-R_2\|\le CW^{-1}\log W$, then
\begin{align}\label{as_Leg} 
\Big|\left\langle t_{00}^{(\ell)}(U)\right\rangle\Big|\le 1- CW^{-1/2}\log^4 W/2.
\end{align}
These relations were proved in Proposition 4.1 of \cite{SS:band_com_d}

Combining (\ref{as_Leg.0})-(\ref{as_Leg.1}) with
 (\ref{k_*.1})  we get, in particular, 
\begin{align*}
\left\langle \sin^{2m}(\theta/2)\right\rangle\le C/W^{2m}
\end{align*}
which gives after some computations (\ref{Kt.0}). 

Relation (\ref{KPsi0}) and (\ref{KPsi.1}) follow from  (\ref{Kt.0}).
To prove (\ref{b_K_R}), let us observe  that in view of  the bounds (\ref{b_rho})
\begin{align}\notag
&\|K_{*R_1,R_2}-K_{R_1,R_2}\|\le CW^{-1/2}\log^2 W\\
\Rightarrow& \mathcal{E}^{(\ell)}K_{R_1,R_2}\mathcal{E}^{(\ell)}\le
\mathcal{E}^{(\ell)}K_{*R_1,R_2}\mathcal{E}^{(\ell)}+CW^{-1/2}\log^2 W.
\label{diff_1}\end{align}
But in view of (\ref{as_Leg.1}) for $\ell>W^{3/4}\log^2W$, we have that
\[
\mathcal{E}^{(\ell)}K_{*R_1,R_2}\mathcal{E}^{(\ell)}\le 1-C'(W^{-1/2}\log^4W)^2+CW^{-1/2}\log^2 W\le  1-C'(W^{-1/2}\log^4W)^2/2.
\]
$\square$

Denote by $P_L$ the orthogonal projection in $\mathcal{H}_0=L_2(\mathcal{H}_{2,+})$  on the space $\mathcal{H}_L$
\begin{align}\notag
\mathcal{H}_L=&\mathrm{Lin}\{\Psi_{*\bar k}(R)\}_{|\bar k|\le L},\quad L=C\log^2W,  \\
\mathcal{P}_L=&P_L\otimes I\Big|_{L_2(U(2))}.
\label{P_L}\end{align}
In what follows we need  to replace $\mathcal{K}_\zeta\to \mathcal{P}_L\mathcal{K}_\zeta\mathcal{P}_L$. For this aim we will use
the following lemma, based, in particular, on (\ref{b_K_R}).
\begin{lemma}\label{l:2} For $L>C\log^2 W$ with sufficiently big $C$
\begin{align}\label{l2.1}
\|(I-\mathcal{P}_L)\mathcal{K}_\zeta(I-\mathcal{P}_L)\|\le (1-C_2L/W).
\end{align}
\end{lemma}
For the proof of the lemma see \cite{SS:band_com_d}.

\medskip

Recall that $\mathcal{K}_0=\mathcal{K}_\zeta\Big|_{\zeta=0}$ and set
\begin{align}\label{bb_K_0}
\mathbb{K}_0=\mathcal{E}_M\mathcal{K}_0\mathcal{E}_M\,\quad M=[L_0 W^{1/4}].
\end{align}
The following lemma provides  information about the eigenvalues and eigenvectors of $\mathbb{K}_0$:
\begin{lemma}\label{l:la_max} 
For any $\ell\le M,$ $\mathbb{K}_0$ has $2\ell+1$ eigenvalues $\lambda_{\ell,k}$ with eigenvectors $\Psi_{\ell,k}(R,U)$  such that
\begin{align}
&|\lambda_{\ell,k}-\lambda_{\max}|\le C(\ell/W)^2,
\label{de_la.0} \end{align}
where $\lambda_\ell$ is defined in (\ref{lam_l}).

Moreover,  there are vectors $h_{\ell,k}\in \mathcal{E}^{(\ell)}$ such that
\begin{align}\label{de_Psi}
&\Psi_{\ell,k}(R,U)=\Psi_{*\bar 0}(R)h_{\ell,k}(U)+O(W^{-1/2})\tilde \Psi_{\ell,k},\quad \tilde \Psi_{\ell,k}\in (1-\mathcal{P}_0)\mathcal{H}.
\end{align}
where $\Psi_{*\bar j}$ defined in (\ref{psi_k}) and $\mathcal{P}_{\bar 0}$ is an orthogonal projection on $\Psi_{*\bar 0}\otimes\mathcal{E}_M$
\end{lemma}
We give here a simpler proof of the lemma than that of Lemma 4.4 in \cite{SS:band_com_d}.

\textit{Proof.} 
Since $\mathbb{K}_0$ has a block diagonal structure with diagonal blocks $\mathbb{K}_\ell=\mathcal{E}^{(\ell)}\mathbb{K}_0\mathcal{E}^{(\ell)}$,
it is sufficient to prove   (\ref{de_la.0}) for $\mathbb{K}_\ell$.  Consider $\mathbb{K}_\ell$ as a $2\times 2$ block matrix, with 
\begin{align*}
&\mathbb{K}_\ell^{(11)}=\widetilde{\mathcal{P}}_{\bar 0}\mathbb{K}_\ell\widetilde{\mathcal{P}}_{\bar 0},\quad 
\mathbb{K}_\ell^{(22)}=(1-\widetilde{\mathcal{P}}_{\bar 0})\mathbb{K}_\ell(1-\widetilde{\mathcal{P}}_{\bar 0}), \\
&\mathbb{K}_\ell^{(12)}=\widetilde{\mathcal{P}}_{\bar 0}\mathbb{K}_\ell(1-\widetilde{\mathcal{P}}_{\bar 0}),\quad \mathbb{K}_\ell^{(21)}
=(1-\widetilde{\mathcal{P}}_{\bar 0})\mathbb{K}_\ell\widetilde{\mathcal{P}}_{\bar 0},
\end{align*}
where $\widetilde{\mathcal{P}}_{\bar 0}$ is an orthogonal projection on $\Psi_{\bar 0}$ defined in (\ref{de_Psi.0}).
Then  (\ref{de_la.0}) follows from the bounds:
\begin{align}\label{de_la.1}
&\|\mathbb{K}_\ell^{(11)}-\lambda_{\max}\|\le C(\ell/W)^2,\quad \|\mathbb{K}_\ell^{(12)}\|\le C(\ell/W)^2,\\
&\mathbb{K}_\ell^{(22)}\le \lambda_{\max}-C/W.
\notag\end{align}
Indeed,   the last two inequalities of (\ref{de_la.1}) imply that for $|z-\lambda_{\max}|\le C(\ell/W)^2$ 
\begin{align}\notag
&\|(\mathbb{K}_\ell^{(22)}-z)^{-1}\|\le C_1W,\\
\Rightarrow&\mathbb{X}_\ell(z)=\mathbb{K}_\ell^{(11)}-z-\mathbb{K}_\ell^{(12)}(\mathbb{K}_\ell^{(22)}-z)^{-1}\mathbb{K}_\ell^{(21)}
=\mathbb{K}_\ell^{(11)}-z+o(W^{-2}).
\label{b_X}\end{align}
Hence, for real $z$ such that $|z-\lambda_{\max}|\le C(\ell/W)^2$ all eigenvalues of $\mathbb{X}_\ell(z)$ differ from corresponding eigenvalues of 
$\mathbb{K}_\ell^{(11)}-z$
less than $cW^{-2}$. Then the first inequality of (\ref{de_la.1}) gives us (\ref{de_la.0}). 
In addition,  one can conclude that $\mathbb{X}_\ell^{-1}(z)$
 exists for real $z$ such that
$C(\ell/W)^2\le |z-\lambda_{\max}|\le c/W$ with some big enough $C$. But by the standard  formula for the  blocks of the inverse matrix
  of $(\mathbb{K}_\ell-z)$, we have that if
 $\mathbb{X}_\ell^{-1}(z)$ and $(\mathbb{K}_\ell^{(22)}-z)^{-1}$ exist, then there exists also $(\mathbb{K}_\ell-z)^{-1}$.
 Hence,  there are no eigenvalues of $\mathbb{K}_\ell$ in the domain
$C(\ell/W)^2\le |z-\lambda_{\max}|\le c/W$.

To prove (\ref{de_Psi}), we consider $\mathbb{E}^{(\ell)}_{(1-c/W,1+c/W)}$ -- a spectral projection  of $\mathbb{K}_\ell$ corresponding to the interval $[1-c/W,1+c/W]$.
We have
\begin{align*}
&\|(1-\widetilde{\mathcal{P}}_{\bar 0})\Psi_{\ell,k}\|=\|(1-\widetilde{\mathcal{P}}_{\bar 0})\mathbb{E}^{(\ell)}_{(1-c/W,1+c/W)}\Psi_{\ell,k}\|\\&\le
\|(1-\widetilde{\mathcal{P}}_{\bar 0})\mathbb{E}^{(\ell)}_{(1-c/W,1+c/W)}\|=\|(1-\widetilde{\mathcal{P}}_{\bar 0})\mathbb{E}^{(\ell)}_{(1-c/W,1+c/W)}(1-\widetilde{\mathcal{P}}_{\bar 0})\|^{1/2}.
\end{align*}
But if  the contour $\mathcal{L}$  is a circle whose diameter is $[1-c/W,1+c/W]$, then 
 $\mathbb{E}^{(\ell)}_{(1-c/W,1+c/W)}$ can be represented in the form
\begin{align*}
\mathbb{E}^{(\ell)}_{(1-c/W,1+c/W)}=-\frac{1}{2\pi i}\oint_{\mathcal{L}} \mathbb{G}_\ell(z) dz,\quad \mathbb{G}_\ell(z)=(\mathbb{K}_\ell-z)^{-1}.
\end{align*}
 Then 
\begin{align}\notag
&(1-\widetilde{\mathcal{P}}_{\bar 0})\mathbb{E}^{(\ell)}_{(1-c/W,1+c/W)}(1-\widetilde{\mathcal{P}}_{\bar 0})
=-\frac{1}{2\pi i}\oint_{\mathcal{L}} \big(\mathbb{G}_\ell^{(22)}(z) dz\\
=&-\frac{1}{2\pi i}\oint_{\mathcal{L}}\Big((\mathbb{K}_\ell^{(22)}-z)^{-1}+
(\mathbb{K}_\ell^{(22)}-z)^{-1}\mathbb{K}_\ell^{(21)}\mathbb{X}_\ell^{-1}(z)\mathbb{K}_\ell^{(12)}(\mathbb{K}_\ell^{(22)}-z)^{-1} \Big)dz\notag\\
\Rightarrow &\|(1-\widetilde{\mathcal{P}}_{\bar 0})\mathbb{E}^{(\ell)}_{(1-c/W,1+c/W)}(1-\widetilde{\mathcal{P}}_{\bar 0})\|\notag\\
&\hskip4cm\le 
\max_{z}\{\|\mathbb{K}_\ell^{(12)}\|^2\|(\mathbb{K}_\ell^{(22)}-z)^{-1}\|^2\|\mathbb{X}_\ell^{-1}\|\} c/W\le C\ell^4W^{-2},
\label{b_pro}\end{align}
where we used that $(\mathbb{K}_\ell^{(22)}-z)^{-1}$ is an analytic function in the circle  $|z-\lambda_{\max}|\le c/W$, hence, the integral of this term  is 0,
 $\|(\mathbb{K}_\ell^{(22)}-z)^{-1}\|\le CW$ for sufficiently small $c$, and  (\ref{b_X}) implies
\[
\|\mathbb{X}_\ell^{-1}\| =\|((\mathbb{K}_\ell^{(11)}-z)^{-1}\big(1+o(W^{-2})(\mathbb{K}_\ell^{(11)}-z)^{-1}\big)^{-1}\|
\le CW. 
\]
Bounds (\ref{b_pro}) and (\ref{de_Psi.0}) imply
\[ \Psi_{\ell,k}(R,U)=\Psi_{\bar 0}(R)h_{\ell,k}(U)+O(\ell^2W^{-1})=\Psi_{*\bar 0}(R)h_{\ell,k}(U)+O(W^{-1/2}).
\]
To finish the proof the lemma, we are left to prove (\ref{de_la.1}).
The first and the second inequalities  of (\ref{de_la.1}) follow from (\ref{KPsi0}).  To prove the third inequality, we use that for 
$\Psi_{h}$ with $\mathrm{\Psi}\in {P}_L\mathcal{H}$, $\|\Psi\|=1$, $h\in\mathcal{E}^{(\ell)},\, \|h\|=1$ (\ref{KPsi.1}) holds, and 
take into account that 
\[  ((1-\widetilde{\mathcal{P}}_{\bar 0})\mathcal{A}(1-\widetilde{\mathcal{P}}_{\bar 0})\Psi,\Psi)\le \max_{\bar j\not=0}\lambda_{\bar j}\le 1-c/W. 
\]
Combining this with (\ref{l2.1}) for $\Psi(R,U)\in (1-\mathcal{P}_L)\mathcal{H}$, we get the third bound of (\ref{de_la.1}).

$\square$

Consider $\mathcal{P}_{\bar j}$ - the orthogonal projection on  $\Psi_{*\bar j}\otimes\mathcal{E}_M$, and set
\begin{align*}
\mathbb{K}_{0,(\bar j,\bar j')}=\mathcal{P}_{\bar j}\mathbb{K}_{0}\mathcal{P}_{\bar j'}.
\end{align*}
An important ingredient of our proof of Theorem \ref{t:new}   will be the following bound on the decay of off-diagonal blocks of $\mathbb{K}_{0}$, valid for
any integer $p$:
\begin{align}\label{exp_dec}
\|\mathbb{K}_{0(\bar j,\bar k)}\|\le \tilde C_p\big(\min\{W^{-3/2},W^{-|\bar j-\bar k|/2}\}+W^{-p-1}\big),\quad \bar j\not=\bar k,\quad \min\{|\bar j|,|\bar k|\}\le L.
\end{align}
For the proof of this expansion we use that $\mathcal{Z}(R_1,R_2)$, $F_\zeta$ of (\ref{A_zeta}) and $K_{R_1,R_2}t^{(\ell)}_{0,k}$ can be expand with respect to
$W^{-1/2}R_1,W^{-1/2}R_2$ and $\rho_1$, $\rho_2$ of (\ref{rho}). Hence, one should estimate
\begin{align*}
I_{\bar j,\bar k}(m,\tilde{p}_\ell,)=\int dR_1dR_2 &\mathcal{A}_*(R_1,R_2)e^{-u_*^2\alpha\Tr R_2^2}P_{\bar j}(R_2)P_{\bar k}(R_1)
\\
&\times p_{\ell}(R_1/W^{1/2},R_2/W^{1/2})\prod_{s=1}^{2m} [R_1,R_2-R_1]_{\alpha_s\beta_s}.
\end{align*}
Here $\mathcal{A}_*(R_1,R_2)$ is defined in (\ref{A_*}), $P_{\bar j}(R_2),P_{\bar k}(R_1)$ are the products of the Hermite polynomials   (see (\ref{psi_k})), and $p_{\ell}(R_1/W^{1/2},R_2/W^{1/2})$ is some uniform polynomials
of degree $\ell$ of $R_1,R_2$ written as in (\ref{expR}). Integrating by parts $2m$ times with respect to $R_2$ and using the recurrent formulas for the Hermite polynomial 
and their derivatives, we conclude that
\begin{align*}
I_{\bar j,\bar k}(m,\tilde{p}_\ell)=&O(W^{-(2m+\ell+1)/2}),\quad |\bar j-\bar k|>2m+\ell,\\
I_{\bar j,\bar k}(m,\tilde{p}_\ell)\le &CW^{-(\ell+2m)/2},\quad |\bar j-\bar k|\le 2m+\ell.
\end{align*}
These relations prove (\ref{exp_dec}).

$\square$

In the next lemma we study the action of $\mathcal{K}_\zeta$ on the vectors of the form
\begin{align}\label{Psi_*}
 \Psi_{\epsilon,h}(R,U)=\Psi(R-\epsilon \mathcal{M}(U))h(U),\quad\Psi\in\mathcal{H}_L,\quad h(U)\in \mathcal{E}_{2M},
\end{align}
 with  $\mathcal{M}(U)$ of (\ref{M(U)}) and $\epsilon=(W/N)^{1/2}$.

Set
\begin{align}\label{Psi_eps}
\Psi_{\ell,k,\epsilon}(R,U)=\Psi_{\ell,k}(R-\epsilon\mathcal{M}(U),U),
\end{align}
where  $\{\Psi_{\ell,k}(R)\}$ are defined in (\ref{de_Psi}).

\begin{lemma}\label{l:shift} 
Given any function of the form (\ref{Psi_*}) for any $h\in \mathcal{E}_M$ ($M=[L_0W^{1/4}]$)
we have
\begin{align}\label{cor:1.0}
(\mathcal{K}_\zeta \Psi_{\epsilon,h})(R_1,U_1)= e^{2\nu(U_1)/N}(\mathcal{K}_0\Psi_{0,h})\big(R_1-\epsilon\mathcal{M}(U_1)\big)
+O(W^{-2}),
\end{align}
where  $\nu$ is defined in (\ref{nu}).

For  functions of the form (\ref{Psi_eps}) with $\ell\le M$
 we have
\begin{align}\label{cor:1.2}
(\mathcal{K}_\zeta \Psi_{\ell,k,\epsilon},\Psi_{\ell',k',\epsilon})=&\delta_{\ell,\ell'}\delta_{k,k'}\lambda_{\ell}+N^{-1}(\nu t^{\ell}_k, t^{\ell'}_{k'})
+O(\epsilon W^{-2})
\end{align}
with $\lambda_\ell$ of (\ref{lam_l}).
 \end{lemma}
Since in the paper \cite{SS:band_com_d} the assertion (\ref{cor:1.2}) was given in a slightly different form (see formula  (4.55)), we give here the proof of the
lemma.

 \textit{Proof of Lemma \ref{l:shift}}.
 Expand $F_\zeta(R_2,U_2)\tilde \Psi(R_2-\epsilon\mathcal{M}(U_2))$  into a series with respect to $\epsilon$ and
 consider $U$ represented  as in (\ref{U}).
Then
  \[\Tr \phi (R)\mathcal{M}(U)=a(R)\cos\theta+\sin\theta (b(R)e^{i\psi}+\bar b(R)e^{-i\psi})
  \]
  Hence, each term of the expansion with respect to $\epsilon$ can be written  in terms of operators $\hat\Phi_1$ 
  and $\hat\Phi_2$ of multiplication by $\cos\theta$  and $\sin\theta $.
 We use the representation $t^{(\ell)}_{0k}$ in terms of the associated Legendre polynomials (see (\ref{as_P})),
and   the recursion formulas  
\begin{align}\label{rec_Leg}
\cos\theta \,P^{(\ell)}_{0k}(\cos\theta)=&c_{\ell,k}P^{(\ell+1)}_{0k}(\cos\theta)+d_{\ell,k}P^{(\ell-1)}_{0k}(\cos\theta),\\
\sin\theta \,P^{(\ell)}_{0k}(\cos\theta)=&c_\ell\big(P^{(\ell+1)}_{0k+1}(\cos\theta)-P^{(\ell-1)}_{0k+1}(\cos\theta)\big).
\notag\end{align}
Here $c_{\ell,k},d_{\ell,k},c_\ell$ are some bounded uniformly in $k,\ell$
coefficients, whose concrete form of is not important for us. 

 Then by (\ref{KPsi0}) we have for any $h\in\mathcal{E}_{M}$
 \begin{align*}
 \int dR_2  B(R_1-R_2) \mathcal{Z}(R_1,R_2)F_0(R_2)\Psi_{0}(R_2)[\Phi_{\alpha},K_{R_1,R_2}]h=O(W^{-1/2}(\ell/W)^{2}),\quad \alpha=1,2,
 \end{align*}
where $[.,.]$ denotes a  commutator. Hence, for operator of multiplication by $F_\zeta(R,U)$ the error term for the commutator
is  $O(\epsilon^s W^{-1/2}(\ell/W)^{2})$.
Notice that zero order with respect to $\epsilon$ term contain $e^{\nu(U_2)/N}$, and the commutator with this term gives us an 
error $O(N^{-1} W^{-1/2}(\ell/W)^{2})$. Thus,
\begin{align}
\label{cor:1.3}
(\mathcal{K}_\zeta \Psi_{\epsilon,h})(R_1,U_1)=& F_\zeta(R_1,U_1))
\int B(R_1-R_2) \mathcal{Z}(R_1,R_2)F_\zeta(R_2,U_1)\\ 
&\times 
\Psi(R_2-\epsilon \mathcal{M}(U_1))(K_{R_1,R_2}h)(U_1)dR_2+O(\epsilon W^{-1/2}(\ell/W)^{2}).
\notag\end{align}
Then we replace $F_\zeta(R_1,U_1)$ by $F_0(R_1-\epsilon\mathcal{M}(U_1))$ with an error $O(\epsilon W^{-3/2})$, 
using that in view of
(\ref{pert}) and   (\ref{ti-f})
\begin{align}\label{f_1}
&F_\zeta(R,U)=F_0(R-\epsilon\mathcal{M}(U))e^{f_1(R,U)},\\
&f_1(R,U)=C_1\nu(U)/N+C_2\epsilon W^{-3/2}\Tr\mathcal{M}(U)R^2\varphi_2(1+R/W^{1/2}),
\notag\end{align}
where $\varphi_2(R)$ is some analytic function, obtained from $\varphi_0(R)$ and $\varphi_1(R)$ of (\ref{ti-f}).

Finally, using (\ref{diff_Z}) and (\ref{diff_K_eps}),  we replace
$\mathcal{Z}(R_1,R_2)$ by $\mathcal{Z}(R_1-\epsilon\mathcal{M}(U_1),R_2-\epsilon\mathcal{M}(U_1))$ with an error $O(\epsilon/W^2)$ 
and $K_{R_1,R_2}$ by $K_{R_1-\epsilon\mathcal{M}(U_1),R_2-\epsilon\mathcal{M}(U_1)}$ with an error $O(\epsilon\log^2W/W^2)$.
Thus, integrating first over $R_2$ and changing $R_2-\epsilon\mathcal{M}(U_1)\to R_2$, we get (\ref{cor:1.0}).

It follows directly from (\ref{cor:1.0}), that
\begin{align}
\mathcal{K}_\zeta \Psi_{\ell,k,\epsilon}(R,U)=&\lambda_{\ell}e^{2\nu(U_1)/N}\Psi_{\ell,k,\epsilon}(R,U)+O(\epsilon W^{-3/2}).
\end{align}
Thus, we need only to check that if we take scalar product of the l.h.s. with $\Psi_{\ell',k',\epsilon}$, then
the  term  of order $O(\epsilon W^{-3/2})$ disappears. We recall that the term appears because of the replacement of $F_\zeta(R,U)$
by $F_0(R-\epsilon\mathcal{M}(U)$ (see (\ref{f_1})). Therefore, its contribution to the scalar product will have the form
\begin{align*}
&\epsilon W^{-3/2}\int\Tr\mathcal{M}(U)R^2\varphi_2(1+R/W^{1/2})\Psi_{\ell,k}(R,U)\Psi_{\ell', k'}(R,U)dRdU\\=
&\epsilon W^{-3/2}\int\Tr\mathcal{M}(U)R^2\varphi_2(1+R/W^{1/2})\Psi^2_{0,0}(R)dRh^{(\ell)}_{k}(U)h^{(\ell')}_{k'}(U)dU+O(\epsilon W^{-2}),
\end{align*}
where  we used (\ref{de_Psi}) to replace $\Psi_{\ell, k}(R,U)$ by $\Psi_{0,0}(R)h^{(\ell)}_{k}(U)+O(W^{-1/2})$ and $\Psi_{\ell', k'}(R,U)$ by 
$\Psi_{00}(R)h^{(\ell')}_{k'}(U)+O(W^{-1/2})$.

In order to compute the last integral, we observe that $\Psi_{0,0}$ is invariant with respect to the 
change $R\to VRV^*$ with any unitary $V$.
Making this change and integrating with respect to $dV$, we obtain for any $\tilde \varphi$ and any matrix
 $\mathcal{M}:\mathcal{M}=\mathcal{M}^*,\,\Tr \mathcal{M}=0$
\begin{align}\label{int_2}
\int dR\Psi_{00}^2(R_1)\Tr \tilde \varphi(R)\mathcal{M}=\int dR\Psi_{00}^2(R)\Tr V\mathcal{M}V^*\tilde \varphi(R)dV=0,
\end{align}
since
\[
\int (V\mathcal{M}V^*)_{\alpha,\beta}dV=0.
\]

$\square$

\section{Proof of Theorem \ref{t:new}}\label{s:pr} 

To prove Theorem \ref{t:new},  we need to use asymptotic relations for $\mathcal{K}_\zeta\Psi_{\bar k}t^{(\ell)}_{0j}$ given in the previous section. 
The main technical problem
here is that  the most advanced asymptotic relations (see, e.g., (\ref{act_D}) and (\ref{KPsi.1})) have good estimates for the remainder terms only under
 restrictions $|\bar k|\le L_0\log^2W$ and $\ell\le L_0W^{1/4}$.
Hence, first we need  to prove that one can replace $\mathcal{K}_\zeta$ with $\mathbb{K}$, where
\begin{align}\label{bbK}
\mathbb{K}=\mathcal{P}_L\mathcal{E}_{2M}\mathcal{K}_\zeta \mathcal{E}_{2M}\mathcal{P}_L, \quad L=[L_0\log^2W], \, M=[L_0W^{1/4}].
\end{align}
This is done by Lemma \ref{l:K_L} the proof of which is based on Proposition \ref{p:CT} and Proposition \ref{p:norm} given below.

\begin{proposition}\label{p:CT} Let the matrix $\mathbb{M}=\mathbb{M}^*$,$\|\mathbb{M}\|\le C$ be constructed from the square matrix blocks  
\[\{\mathbb{M}_{kj}\}_{k,j=1}^{N_2},\quad \mathbb{M}_{kj}=P_k\mathbb{M}P_j,\quad P_kP_j=\delta_{kj}P_k,\quad P_k=P_k^*,\quad\sum P_{k}=1.\]
 Consider $\mathbb{M}$ as a $3\times 3$ block matrix with diagonal blocks
\[ \mathbb{M}^{(00)}=\{\mathbb{M}_{kj}\}_{k,j=1}^{N_0},\quad \mathbb{M}^{(11)}=\{\mathbb{M}_{kj}\}_{k,j=N_0+1}^{N_1},
\quad \mathbb{M}^{(22)}=\{\mathbb{M}_{kj}\}_{k,j=N_1+1}^{N_2}.
\]
Assume, that the following conditions are true
\begin{enumerate}
\item[(i)] $N_1-N_0>c\log^2N$;
\item[(ii)] $\mathbb{M}_{kj}=0$ for $|j-k|>p_0$, if $\min\{j,k\}\le N_1$;
\item[(iii)] $D:=\mathrm{diag}\{\mathbb{M}^{(11)}_{kk}\}_{k=N_0}^{N_1}<-\mathcal{C}_1<0$,$\qquad$ $\|\mathbb{M}^{(11)}-D\|\le q\mathcal{C}_1/2,\quad q<1$; 
\item[(iv)] $ \mathbb{M}^{(22)}\le -\mathcal{C}_1$, $\quad$ $\|\mathbb{M}^{(12)}\|^2<q'(1-q)(\mathcal{C}_1/2)^{2}$, $q'<1$;
\end{enumerate}
where $\mathcal{C}_{1}=C_1N^{-\alpha_1}$  with fixed $C_1$ and $\alpha_1\ge 0$. 

Then,
 denoting  $\mathfrak{M}_1$  the matrix constructed from the blocks $\mathbb{M}^{(\alpha\beta)},(\alpha,\beta=0,1)$, we have
\begin{align}\label{CT.1}
 &\mathbb{M}\le  \lambda_{\max}(\mathfrak{M}_1)+\delta_0^{1/2}, \quad \delta_0=q^{|N_1-N_0|/p_0}
 \end{align}
 If $ \lambda_{\max}(\mathfrak{M}_1)\le k_0/N$,  $\|g\|\le N^m$,  and $P_kg=0,\, (\forall  k>N_0)$, then
\begin{align}\label{CT.1'}
&((1+\mathbb{M})^{N-1}g,g)=((1+\mathfrak{M}_1)^{N-1}g,g)+O(\delta_0^{1/2}), 
\end{align}
In addition, if we denote by  $E_{\lambda}(\mathbb{M})$ the spectral projection of $\mathbb{M}$ on the interval $\lambda'>\lambda$
\begin{align} 
&\|E_{-\mathcal{C}_1/2}(\mathbb{M})P_k\|\le C(\mathcal{C}_1(1-q))^{-1}\|\mathbb{M}^{(10)}\|q^{k-N_0}+O(\delta_0^{1/2}),\quad k>N_0.
\label{CT.0}\end{align}
\end{proposition}

\textit{Proof.}
It follows from  (iii) that for any $z>-\mathcal{C}_1/2$  
\[
z-D_{kk}=z-\mathbb{M}_{kk}>\mathcal{C}_1/2\Rightarrow \|(\mathbb{M}^{(11)}-D)(z-D)^{-1}\|<q.
\] 
 Taking into account that the matrix $B:=(\mathbb{M}^{(11)}-D)(z-D)^{-1}$ has non-zero entries only on $2p_0-1$ central diagonals, we 
 conclude that  $(B^s)_{jk}=0$, if $s<|j-k|/p_0$, and so
\begin{align}\notag
&\|(\mathbb{M}^{(11)}-z)^{-1}_{jk}\|=\Big\|(z-D)^{-1}_{jj}\Big(1-(\mathbb{M}^{(11)}-D) (z-D)^{-1}\Big)^{-1}_{jk}\Big\|\\
&\le 2\mathcal{C}_1^{-1}\sum_{s \ge |j-k|/p_0}q^s\le 2(\mathcal{C}_1(1-q))^{-1}q^{|j-k|/p_0},
\label{CT.2}\\
\notag
& z-\mathbb{M}^{(11)}=z-D-(\mathbb{M}^{(11)}-D)\ge  \mathcal{C}_1/2-q\mathcal{C}_1/2
 = \mathcal{C}_1(1-q)/2.
\end{align}
Since by (i) $\mathbb{M}^{(10)}$ and $\mathbb{M}^{(21)}$ have non-zero entries only in the upper right $p_0\times p_0$ corner, one can conclude
that $\mathbb{M}^{(21)}(\mathbb{M}^{(11)}-z)^{-1}\mathbb{M}^{(10)}$ has non-zero entries only in the upper right $p_0\times p_0$ corner and 
only  $(\mathbb{M}^{(11)}-z)^{-1}_{jk}$ with $|j-k|>N_1-N_0-2p_0$ give non-zero contribution to $\mathbb{M}^{(21)}(\mathbb{M}^{(11)}-z)^{-1}\mathbb{M}^{(10)}$. Hence, using (\ref{CT.2}) and (iv), we get
\begin{align}\label{CT.2'}
\|\mathbb{M}^{(21)}(\mathbb{M}^{(11)}-z)^{-1}\mathbb{M}^{(10)}\|\le
 p_0^2\|\mathbb{M}^{(21)}\|\|\mathbb{M}^{(10)}\|\frac{2q^{|N_1-N_0|/p_0-2}}{\mathcal{C}_1(1-q)}=O(\delta_0N^{\alpha}).
\end{align}
Here and below in the proof of the proposition we do not control  the exact degree of  $N^{\alpha}$ in our bounds, since
 $N^{\alpha}\delta_0=O(\delta_0^{1/2})$ for any fixed $\alpha$.
 
Denote by $\mathfrak{M}_2$ the matrix constructed from blocks $\mathbb{M}^{(\alpha\beta)},(\alpha,\beta=1,2)$, and by \\ $\widehat{\mathfrak{M}}_2=\mathrm{diag}\{\mathbb{M}^{(11)},\mathbb{M}^{(22)}\}$.
 Matrix 
$\mathfrak{M}_2-z$ is invertible for any $z\ge -\mathcal{C}_1/2$, since by condition  (iv) and above bounds we have
\begin{align}\notag
&\mathbb{M}^{(11)}-z-\mathbb{M}^{(12)}(\mathbb{M}^{(22)}-z)^{-1}\mathbb{M}^{(21)}\le
 -\mathcal{C}_1(1-q)/2+\|\mathbb{M}^{(12)}\|^2(\mathcal{C}_1/2)^{-1}
\\ \le& -\mathcal{C}_1(1-q)(1-q')/2
\Rightarrow \|\big((\mathfrak{M}_2-z\big)^{-1})^{(12)}\|=O(N^\alpha).
\label{CT.3}\end{align}
Moreover, using the resolvent identity and  (\ref{CT.2'}), (\ref{CT.3}), we get
\begin{align}\notag
\Big\|\Big(\big(\mathfrak{M}_2-z)^{-1}\big)^{(11)}-&(\mathbb{M}^{(11)}-z)^{-1}\Big)\mathbb{M}^{(10)}\Big\|=
\Big\|\Big((\mathfrak{M}_2-z)^{-1}-(\widehat{\mathfrak{M}}_2-z)^{-1}\Big)^{(11)}\mathbb{M}^{(10)}\Big\|\\
=&\|((\mathfrak{M}_2-z)^{-1})^{(12)}\mathbb{M}^{(21)}(\mathbb{M}^{(11)}-z)^{-1}\mathbb{M}^{(10)}\|=O(N^\alpha\delta_0).
\label{CT.4}\end{align}
To prove (\ref{CT.0}), consider $\Phi_\lambda$ -- an eigenvector of $\mathbb{M}$ corresponding to $\lambda\ge -\mathcal{C}_1/2$.
Set 
\[ \Phi_{0,\lambda}=\{P_j\Phi_\lambda\}_{j=1}^{N_0},\quad \Phi_{1,\lambda}=\{P_j\Phi_\lambda\}_{j=N_0+1}^{N_1},\quad
\Phi_{2,\lambda}=\{P_j\Phi_\lambda\}_{j=N_1+1}^{N_2},\quad \Phi_{2,\lambda}'=\left(\begin{array}{c}\Phi_{1,\lambda}\\ \Phi_{2,\lambda}\end{array}\right).
\]
Using the eigenvalue equation, we get 
\[
\Phi_{2,\lambda}'=-(\mathfrak{M}_2-\lambda)^{-1}\left(\begin{array}{c}\mathbb{M}^{(10)}\Phi_{0,\lambda}\\ 0\end{array}\right)=
-\left(\begin{array}{c}((\mathfrak{M}_2-\lambda)^{-1})^{(11)}\mathbb{M}^{(10)}\Phi_{0,\lambda}\\
 ((\mathfrak{M}_2-\lambda)^{-1})^{(21)}\mathbb{M}^{(10)}\Phi_{0,\lambda}\end{array}\right).
\]
But
\[ ((\mathfrak{M}_2-\lambda)^{-1})^{(21)}=- (\mathbb{M}^{(22)}-\lambda)^{-1}\mathbb{M}^{(21)}((\mathfrak{M}_2-\lambda)^{-1})^{(11)}.
\]
Hence,  (\ref{CT.2'}) and  (\ref{CT.4}) yield
\[
\Phi_{2,\lambda}'=-\left(\begin{array}{c}(\mathbb{M}^{(11)}-\lambda)^{-1}\mathbb{M}^{(10)}\Phi_{0,\lambda}\\ 0\end{array}\right)+O(\delta_0^{1/2})
\]
which implies (\ref{CT.0}), if we use (\ref{CT.2}).

\medskip

To prove (\ref{CT.1}), we set for $z>\lambda_{\max}(\mathfrak{M}_1)$
\begin{align*}
&\mathcal{M}(z)=\mathbb{M}^{(00)}-z-\mathbb{M}^{(01)}((\mathfrak{M}_2-z)^{-1})^{(11)}\mathbb{M}^{(10)},\\
&\mathcal{M}_0(z)=\mathbb{M}^{(00)}-z-\mathbb{M}^{(01)}(\mathbb{M}^{(11)}-z)^{-1}\mathbb{M}^{(10)}.
\end{align*}
By (\ref{CT.4}), we have
\begin{align}\label{CT.4'}
\mathcal{M}(z) -\mathcal{M}_0(z)= O(N^{\alpha}\delta_0).
\end{align}
Since  $\mathfrak{M}_1$ is invertible for any $z\ge \lambda_{\max}(\mathfrak{M}_1)$, we conclude that  $\mathcal{M}_0(z)$ is invertible for any $z\ge \lambda_*$.
Hence  
\[\mathcal{M}_0(\lambda_{\max}(\mathfrak{M}_1))<0,\] 
(otherwise, since
$\mathcal{M}_0(z)\to -\infty$, as $z\to \infty$, we get that $\mathcal{M}_0(z)$ has zero eigenvalue for some $z\ge \lambda_{\max}(\mathfrak{M}_1)$).
Then, taking into account that for any $z\ge \lambda_{\max}(\mathfrak{M}_1),\delta>0$
\begin{align*}
-\mathbb{M}^{(01)}(\mathbb{M}^{(11)}-z-\delta)^{-1}\mathbb{M}^{(10)}<-\mathbb{M}^{(01)}(\mathbb{M}^{(11)}-z)^{-1}\mathbb{M}^{(10)},
\end{align*}
we obtain  for any $z>\lambda_{\max}(\mathfrak{M}_1)+\delta_0^{1/2}$ 
\begin{align*}
\mathcal{M}(z) \le&\mathcal{M}_0(z)+O(N^\alpha\delta_0) \le \mathcal{M}_0(z)= \mathbb{M}^{(00)}-z+\delta_0^{1/2}-
\mathbb{M}^{(01)}(\mathbb{M}^{(11)}-z)^{-1}\mathbb{M}^{(10)}\\
\le &\mathbb{M}^{(00)}-\lambda_{\max}(\mathfrak{M}_1)-\mathbb{M}^{(01)}(\mathbb{M}^{(11)}-\lambda_{\max}(\mathfrak{M}_1))^{-1}\mathbb{M}^{(10)}
=\mathcal{M}_0(\lambda_{\max}(\mathfrak{M}_1))<0.
\end{align*}
Thus we have proved  (\ref{CT.1}). To prove (\ref{CT.1'}), we use (\ref{CT.4'}) to write
\begin{align*}
(\mathbb{M}^{N-1}g,g)=&-\frac{1}{2\pi i}\oint_{|z|=1+2k_0/N}z^{N-1}((\mathbb{M}-z)^{-1})^{(00)}g,g)\\=&
-\frac{1}{2\pi i}\oint_{|z|=1+2k_0/N}z^{N-1}(\mathcal{M}^{-1}(z)g,g)\\=&-\frac{1}{2\pi i}\oint_{|z|=1+2k_0/N}z^{N-1}(\mathcal{M}_0^{-1}(z)g,g)+O(\delta_0^{1/2})=
(\mathfrak{M}_1^{N-1}g,g)+O(\delta_0^{1/2}).
\end{align*}

$\square$

\begin{proposition}\label{p:norm} Consider a $2\times 2$  block matrix $\widetilde{\mathbb{M}}=\widetilde{\mathbb{M}}^*$ with 
blocks $\widetilde{\mathbb{M}}^{(\alpha\beta)}$, such that
\[ \widetilde{\mathbb{M}}^{(11)}<m_1,\quad \widetilde{\mathbb{M}}^{(22)}<m_2.\]
If $m_1\not=m_2$, then
\begin{align}\label{p:n.1}
\lambda_{\max}(\widetilde{\mathbb{M}})\le  \lambda^*=\max\{m_1,m_2\}
+\|\widetilde{\mathbb{M}}^{(12)}\|^2|m_2-m_1|^{-1}.
\end{align}
\end{proposition}

\textit{Proof of Proposition \ref{p:norm}}. Assume that $m_1>m_2$. Then bound (\ref{p:n.1}) follows from the inequality valid for any $\lambda> \lambda_*$ :
\[
\widetilde{\mathbb{M}}^{(11)}-\lambda-\widetilde{\mathbb{M}}^{(12)}(\widetilde{\mathbb{M}}^{(22)}-\lambda)^{-1}\widetilde{\mathbb{M}}^{(21)}\le m_1-\lambda
+\|\widetilde{\mathbb{M}}^{(12)}\|^2|m_2-m_1|^{-1}=\lambda_*-\lambda<0
\]
Hence, the matrix in the l.h.s. is invertible, and since $\widetilde{\mathbb{M}}^{(22)}-\lambda$ is also invertible, using the inverse formula for the block matrices, we conclude  that $\widetilde{\mathbb{M}}-\lambda$ is invertible for 
$\lambda> \lambda_*$.

If $m_2>m_1$ we can use that $\widetilde{\mathbb{M}}-\lambda$ is invertible if
there exist inverse for the matrices   
\[\widetilde{\mathbb{M}}^{(11)}-\lambda\quad \hbox{and}\quad
\widetilde{\mathbb{M}}^{(22)}-\lambda-\widetilde{\mathbb{M}}^{(21)}(\widetilde{\mathbb{M}}^{(11)}-\lambda)^{-1}\widetilde{\mathbb{M}}^{(12)},
\]
and then obtain the same bound.
$\square$

\begin{lemma}\label{l:K_L}  Let $\mathbb{K}$ be defined in (\ref{bbK}). Then there is a finite $s>0$ such that
\begin{align}\label{K_L.1}
\tilde\Theta(z_1,z_2)=&(\mathbb{K}^{N-1}g_s,g_s)+o(1),\quad g_s=\mathcal{P}_{\le s}g,\quad \mathcal{P}_{\le s}:=\sum_{|k|\le s}\mathcal{P}_{\bar k},
\end{align}
where $\mathcal{P}_{\bar k}$ is an orthogonal projection on $\Psi_{*\bar k}$.
\end{lemma}
\textit{Proof.} We start from the proof of the inequality
\begin{align}\label{n.0}
\|\mathcal{K}_\zeta\|\le 1+k/N.
\end{align}
Recall that the operator of multiplication
by $F_\zeta(R,U)$ has the form (\ref{f_1}). Observe that the remainder in (\ref{f_1}) looks like
\[
F_0(R-\epsilon\mathcal{M})(1+O(N^{-1}))
\]
Hence, it is sufficient to prove (\ref{n.1}) for  the operator 
$\widetilde{\mathcal{K}}_{\zeta}$ which  corresponds to $\mathcal{K}_{\zeta}$ with  $F_\zeta(R,U)$ replaced by $F_0(R-\epsilon\mathcal{M})$.

We consider $M'=[M_0W^{3/4}\log W ]$, $L=[M_0\log^2W]$, $M=[L_0W^{1/4}]$  with sufficiently big fixed $M_0,L_0$, and use Proposition \ref{p:CT} to prove the
inequalities
\begin{align}\notag
&\|\widetilde{\mathcal{K}}_\zeta\|\le \|\mathcal{E}_{2M'}\widetilde{\mathcal{K}}_\zeta\mathcal{E}_{2M'}\|+1/N\le 
 \|\mathcal{P}_L\mathcal{E}_{2M'}\widetilde{\mathcal{K}}\mathcal{E}_{2M'}\mathcal{P}_L\|+2/N\\
 &\le  
 \|\mathcal{P}_L\mathcal{E}_{2M}\mathcal{K}_\zeta\mathcal{E}_{2M}\mathcal{P}_L\|+3/N\le 1+k/N.
\label{n.1}\end{align}
To use  Proposition \ref{p:CT}  for the proof of the first inequality, we notice that if $\zeta=0$, then for each $\ell=0,1,\dots$ the space
$\mathcal{H}\otimes\mathcal{E}^{(\ell)}$ is invariant with respect to $\mathcal{K}_{0}$.
Moreover,  since multiplication by $f_1$ can transform $h\in \mathcal{E}^{(\ell)}$ into a function which has non-zero components only in 
$\mathcal{E}^{(\ell-1)},\mathcal{E}^{(\ell)},\mathcal{E}^{(\ell+1)}$, the matrix $\widetilde{\mathcal{K}}_{\zeta,L}$ 
is ``block three-diagonal"  in the basis of $\mathcal{H}_{L}\otimes\mathcal{E}^{(\ell)}$. 
Set 
\[
\widetilde{\mathcal{K}}_{\zeta}^{(\ell\ell')}=\mathcal{E}^{(\ell)}\widetilde{\mathcal{K}}_{\zeta}\mathcal{E}^{(\ell')},
\]
We would like to apply Proposition \ref{p:CT} to the matrix $\widetilde{\mathcal{K}}_{\zeta}-I$ written as a three diagonal matrix with entries
$\widetilde{\mathcal{K}}_{\zeta}^{(\ell\ell')}$ and   $N_0=M'$, $N_1=2M'$ and $N_2=\infty$.
Assume that we prove the inequality
\begin{align}\label{b_K_ll}
\widetilde{\mathcal{K}}_{\zeta}^{(\ell\ell)}\le 1-C_1L/W,\quad \ell>M'.
\end{align}
Then, since by the forms of the operators $\widetilde{\mathcal{K}}_{\zeta}$ and $\widetilde{\mathcal{K}}_{0}$ (see (\ref{A}), (\ref{pert}), and (\ref{ti-f}))
\begin{align}\label{K_L.2}
\widetilde{\mathcal{K}}_{\zeta}=\widetilde{\mathcal{K}}_{0}+O(\varepsilon/W)\Rightarrow \| \widetilde{\mathcal{K}}_{\zeta}^{(\ell\ell\pm 1)}\|\le C\epsilon/W,
\end{align}
one can check easily that conditions (i)-(v) are fulfilled, since
for the matrix $\mathbb{M}=\widetilde{\mathcal{K}}_{\zeta}-I$ with \\ $\mathcal{C}_1=-C_1L/W$, $q=\epsilon$,  $q'=\epsilon^2$ we have
\[
2\|\widetilde{\mathcal{K}}_{\zeta}^{(\ell\ell+1)}\|\mathcal{C}_1^{-1}\le C\epsilon/\log W\le \epsilon=q,\quad
\|\widetilde{\mathcal{K}}_{\zeta}^{(\ell\ell+1)}\|^2(2/C_1)^2\le C(\epsilon/\log W)^2(1-q)^{-1}\le \epsilon^2=q'.
\]

To prove (\ref{b_K_ll}), we
apply Proposition \ref{p:norm} to the matrix $\widetilde{\mathbb{M}}_{\ell}=\widetilde{\mathcal{K}}_{0}^{(\ell\ell)}$  considered as a block matrix 
with 
\[
\widetilde{\mathbb{M}}^{(11)}_{\ell}=\mathcal{P}_L\widetilde{\mathbb{M}}_{\ell}\mathcal{P}_L,\quad\widetilde{\mathbb{M}}_{\ell}^{(12)}=
\mathcal{P}_L\widetilde{\mathbb{M}}_{\ell}(1-\mathcal{P}_L),
\quad\widetilde{\mathbb{M}}_{\ell}^{(22)}=(1-\mathcal{P}_L)\widetilde{\mathbb{M}}_{\ell}(1-\mathcal{P}_L)
\]
with $\mathcal{P}_L$ of (\ref{P_L}). We use bounds 
\begin{align*}
&\widetilde{\mathbb{M}}_\ell^{(11)}\le 1-C_1W^{-1/2}\log^2W,\\
 & \|\widetilde{\mathbb{M}}_\ell^{(12)}\|\le CW^{-3/2}\log^3 W,
\quad\widetilde{\mathbb{M}}_\ell^{(22)}\le 1-CL/W,
\notag\end{align*}
where the first one follows from (\ref{KPsi.1}) and (\ref{b_K_R}), the second -- from (\ref{exp_dec}), and the last one -- from Lemma \ref{l:2}. 
Hence we get  (\ref{b_K_ll}) for $\widetilde{\mathcal{K}}_{0}^{(\ell\ell)}$. Then (\ref{K_L.2}) implies (\ref{b_K_ll}) for $\widetilde{\mathcal{K}}_{\zeta}^{(\ell\ell)}$.

\medskip

At the second step  we  use Proposition \ref{p:CT} in order to prove the second inequality in (\ref{n.1}).
We consider $\widetilde{\mathcal{K}}_\zeta-I$ as a matrix constructed from blocks
\[
\widetilde{\mathcal{K}}_{\zeta,\bar k\bar k'}=\mathcal{P}_{\bar k}\mathcal{E}_{2M'}\widetilde{\mathcal{K}}_\zeta\mathcal{E}_{2M'}\mathcal{P}_{\bar k'}.
\]
Using an expansion  (\ref{exp_dec}) $\widetilde{\mathcal{K}}_\zeta$ and expanding $F_\zeta-F_0$ with respect to $\zeta$ one can get
the analogue of (\ref{exp_dec}) for $\widetilde{\mathcal{K}}_\zeta$, valid for any integer $p$
\begin{align}\label{exp_K.1}
\|\widetilde{\mathcal{K}}_{\zeta,\bar k\bar k'}\|\le \tilde C_p\big(\min\{\epsilon/W,W^{-|\bar j-\bar k|/2}\}+W^{-p-1}\big),\quad 
\bar j\not=\bar k,\quad \min\{|\bar j|,|\bar k|\}\le L.
\end{align}
Hence $\mathcal{E}_{2M'}\widetilde{\mathcal{K}}_\zeta\mathcal{E}_{2M'}$ can be written as 
\[\mathcal{E}_{2M'}\widetilde{\mathcal{K}}_\zeta\mathcal{E}_{2M'}=\mathcal{E}_{2M'}\widetilde{\mathcal{K}}_\zeta'\mathcal{E}_{2M'}+O(N^{-3}),
\] 
where  matrix $\mathcal{E}_{2M'}\widetilde{\mathcal{K}}'_\zeta\mathcal{E}_{2M'}\mathcal{P}_L$ has a finite number of non-zero  diagonals. 
Thus conditions of Proposition \ref{p:CT} are fulfilled for
$\mathbb{M}=\mathcal{E}_{2M'}\widetilde{\mathcal{K}}'_\zeta\mathcal{E}_{2M'}$ 
and we can take $N_0=1$, $N_1=L$, 
 $\mathcal{C}_1=C_1/W$, $q=C_2\epsilon$ with  sufficiently big $C_2$,  $q'=C_2\epsilon^2$, since (\ref{exp_K.1})  yields
 \[
\|\mathbb{M}^{(11)}-D\|\le C\epsilon/W,\quad \mathbb{M}^{(12)}\le -C\epsilon/W. 
 \]

Finally, we apply Proposition \ref{p:CT}  at the third time to the matrix 
$\mathbb{M}=\mathcal{P}_L\mathcal{E}_{2M'}(\widetilde{\mathcal{K}}_\zeta-I)\mathcal{E}_{2M'}\mathcal{P}_L$ in order to replace $M'=[M_0W^{3/4}]$
by $M=[L_0 W^{1/4}]$. Then we take $N_0=M$, $N_1=2M$, $q=q'=\frac{1}{2}$ and recall that by  for $\ell<2M$ by (\ref{KPsi.1}) and (\ref{lam_l})
\[
\mathcal{P}_L\widetilde{\mathcal{K}}_\zeta^{\ell\ell}\mathcal{P}_L\le 1-(2u_*)^{-2}(\ell/W)^2\le 1-CL_0^2W^{-3/2},\quad
\|\mathcal{P}_L\widetilde{\mathcal{K}}_\zeta^{\ell\ell+1}\mathcal{P}_L\|\le C\epsilon/W=CW^{-3/2}.
\]
Now to finish the proof of (\ref{n.0}) it suffices to check that
\[
\mathbb{K}\le 1+k_0/N.\]
The last inequality can be obtained from Proposition \ref{p:norm} if we consider $\widetilde{\mathbb{M}}=\mathbb{K}$
as a block matrix with diagonal blocks
\begin{align*}
\widetilde{\mathbb{M}}^{(11)}
=\mathcal{P}_{\bar 0}\widetilde{\mathbb{M}}\mathcal{P}_{\bar 0},\quad \widetilde{\mathbb{M}}^{(22)}
=(I-\mathcal{P}_{\bar 0})\widetilde{\mathbb{M}}(1-\mathcal{P}_{\bar 0}),
\end{align*}
where $\mathcal{P}_{\bar 0}$ is an orthogonal projection on $\Psi_{*\bar 0}$ and prove the bounds
\begin{align*}
\widetilde{\mathbb{M}}^{(11)}\le 1+k_1/N,\quad
\| \widetilde{\mathbb{M}}^{(12)}\|\le C\epsilon/W=C'W^{-3/2},\quad
\widetilde{\mathbb{M}}^{(22)}\le 1-C/W.
\end{align*}
The first inequality  follows from (\ref{KPsi0}) and (\ref{diff_eig}).
The second and the third inequalities follows from
(\ref{exp_K.1}) and (\ref{KPsi.1}), (\ref{diff_eig}). Thus we have finished the proof of (\ref{n.1}).

\medskip

The next step is to use (\ref{CT.1'}) to obtain (\ref{K_L.1}) with $g_s$ replaced
by $g$. For $\mathcal{E}_{2M'}\widetilde{\mathcal{K}}_\zeta\mathcal{E}_{2M'}$ we can do this directly, since conditions of (\ref{CT.1})
are obviously valid. But  for replacement $\mathcal{E}_{2M'}\widetilde{\mathcal{K}}_\zeta\mathcal{E}_{2M'}$ with 
$\mathcal{P}_L\mathcal{E}_{2M'}\widetilde{\mathcal{K}}_\zeta\mathcal{E}_{2M'}\mathcal{P}_L$, conditions of (\ref{CT.1}) are not fulfilled
initially, since $g$ has non-zero components for $|\bar k|> N_0=1$. Hence, 
we apply (\ref{CT.0}) first in order to show that there is a fixed $s$ such that one can replace $g$ by $g_s$.
Since $q=C_2\epsilon=C_2'W^{-1/2}$,  (\ref{CT.0}) implies  that there exists  fixed $s>0$ such that
\[\|(1-P_s)E_{1-C_1/2W}^{(M')}(\mathcal{E}_{2M'}\widetilde{\mathcal{K}}_\zeta\mathcal{E}_{2M'})\|\le W^{-3}.
\]
where $E^{(M')}_\lambda:=E_\lambda(\mathcal{E}_{2M'}\widetilde{\mathcal{K}}_\zeta\mathcal{E}_{2M'})$ is a spectral projection of 
$\mathcal{E}_{2M'}\widetilde{\mathcal{K}}_\zeta\mathcal{E}_{2M'}$. Then
 \begin{align*}
((\mathcal{E}_{2M'}&\widetilde{\mathcal{K}}_\zeta\mathcal{E}_{2M'})^{N-1}g,g)=\Big(\int_{\lambda>1-C_1/2W}+\int_{\lambda\le 1-C_1/2W}\Big) \lambda^{N-1}
d(E^{(M')}_\lambda g,g)\\=&
\int_{\lambda>1-C_1/2W}\lambda^{N-1}d(E^{(M')}_\lambda (P_s+1-P_s)g,(P_s+1-P_s)g)+O(e^{-cN/W})\\
=&((\mathcal{E}_{2M'}\widetilde{\mathcal{K}}_\zeta\mathcal{E}_{2M'})^{N-1} P_s g,P_s g)+O(\|(1-P_s)E_{1-C_1/2W}^{(M')}\|\|g\|^2)\\ =&
((\mathcal{E}_{2M'}\widetilde{\mathcal{K}}_\zeta\mathcal{E}_{2M'})^{N-1} g_s,g_s )+O(W^{-1}).
\end{align*}  
Then choosing $N_0=1+s$
we can apply the second line of (\ref{CT.1}) to obtain the analogue  of (\ref{K_L.1}) with $\mathbb{K}$ replaced by 
$\mathcal{P}_L\mathcal{E}_{2M'}\widetilde{\mathcal{K}}_\zeta\mathcal{E}_{2M'}\mathcal{P}_L$ and $g$ replaced by $g_s$. At the last step,
going from $\mathcal{P}_L\mathcal{E}_{2M'}\widetilde{\mathcal{K}}_\zeta\mathcal{E}_{2M'}\mathcal{P}_L$  to $\mathbb{K}$, one can use
the second line of (\ref{CT.1})  directly.

$\square$

\begin{lemma}\label{l:P} Denote
$\mathfrak{P}_{\epsilon}^{(1)}$ the orthogonal projection on the subspace $\mathrm{Lin}\{\Psi_{\ell,k,\epsilon}\}_{\ell\le M,|k|\le \ell}$
 defined by (\ref{Psi_eps}) for $\Psi_{\ell,k}$ of (\ref{de_Psi}). 
  Then
\begin{align}\label{l:P.1}
(\mathbb{K}^{N-1}g_s,g_s)=
((\mathfrak{P}_{\epsilon}^{(1)}\mathbb{K}\mathfrak{P}_{\epsilon}^{(1)})^{N-1} g_0,g_0)+o(1),\quad g_0=\mathcal{P}_{\bar 0}g.
\end{align}
\end{lemma}

\textit{Proof of Lemma \ref{l:P}.} 
 
 Set  $\mathfrak{P}_{\epsilon}^{(2)}=I-\mathfrak{P}_{\epsilon}^{(1)}$ and
\begin{align}\label{K^ab.0}
\mathbb{K}^{(\alpha\beta)}=&\mathfrak{P}_{\epsilon}^{(\alpha)}\mathbb{K}\mathfrak{P}_{\epsilon}^{(\beta)},\quad \alpha,\beta=1,2.
\end{align}
Then we have  the bounds
\begin{align}\label{b_K.0}
&\| \mathbb{K}^{(11)}-\lambda_{\max} \|\le C(\epsilon/W+N^{-1}),\\
 & \|\mathbb{K}^{(12)}\|\le C \epsilon W^{-3/2},\quad
\mathbb{K}^{(22)}\le 1-C_0/W.
\notag\end{align}
The first bound here follow from Lemma \ref{l:la_max} and (\ref{cor:1.2}), the second -- from (\ref{cor:1.0}), and the last  bound was proven in Lemma \ref{l:K_L}.

Since  we proved in  Lemma \ref{l:la_max} that $\mathbb{K}_0=\mathbb{K}\Big|_{\zeta=0}$ has $(M+1)^2$ eigenvalues in the $(\epsilon/W)$-neighbourhood of 
$\lambda_{\max}$ and all
the remaining eigenvalues are less that $1-W$ and
\begin{align}\label{diff_K}
\|\mathbb{K}-\mathbb{K}_{0}\|\le C\epsilon/W, 
\end{align}
we conclude that the same property is valid for $\mathbb{K}$. Hence, we can apply
 the Cauchy residue theorem in the following  form:
\begin{align*}
(\mathbb{K}^{N-1} g_s,g_s)=&-\frac{1}{2\pi i}\Big(\oint_{\mathcal{L}}-\oint_{|z|=1-c/2W}\Big)
z^{N-1}\sum_{\alpha,\beta=1}^2\Big(\mathbb{G}^{(\alpha\beta)}(z)\mathfrak{P}_{\epsilon}^{(\beta)}g_s,\mathfrak{P}_{\epsilon}^{(\alpha)}g_s\Big)dz\\
=&-\frac{1}{2\pi i}\oint_{\mathcal{L}}z^{N-1}
\sum_{\alpha,\beta=1}^2\Big(\mathbb{G}^{(\alpha\beta)}(z)\mathfrak{P}_{\epsilon}^{(\beta)}g_s,\mathfrak{P}_{\epsilon}^{(\alpha)}g_s\Big)dz+o(\lambda_{\max}^{N-1}).\end{align*}
Here $\mathbb{G}(z)=(\mathbb{K}-z)^{-1}$ and
\begin{align*}
\mathcal{L}=&\partial \Omega,\quad \Omega=\{z:|z|\le \lambda_{\max}(1+2k_0/N)\wedge|z-\lambda_{\max}|\le C(\epsilon/W)\}.
\end{align*}
Since the spectrum of $\mathbb{K}$ belongs to $[0,\lambda_{\max}(1+k_0/N)]$ (see (\ref{n.1})), by (\ref{b_K.0}) and the standard resolvent bounds we have for $z\in \mathcal{L}$
\begin{align*}
&\|(\mathbb{K}^{(11)}-z)^{-1}\| \le C|z-\lambda_{\max}(1+k_0/N)|^{-1},\\
&\|\mathbb{G}^{(11)}(z)\|\le C|z-\lambda_{\max}(1+k_0/N)|^{-1},
\qquad \qquad\|\mathbb{G}^{(22)}\|\le CW,\\
&\|\mathbb{G}^{(12)}\|= \|(\mathbb{K}^{(11)}-z)^{-1}\mathbb{K}^{(12)}\mathbb{G}^{(22)}\|\le C|z-\lambda_{\max}(1+k_0/N)|^{-1}\epsilon/W^{1/2}.
\end{align*}
Therefore, we conclude that the integrals with $\mathbb{G}^{(12)}$ and $\mathbb{G}^{(21)}$ gives us $o(\lambda_{\max}^{N-1})$.
In addition, using (\ref{b_K.0}) and the last bound, we obtain
\begin{align*}
&\Big|\oint_{\mathcal{L}} z^{N-1}
\Big(\big(\mathbb{G}^{(11)}(z)-(\mathbb{K}^{(11)}-z)^{-1}\big)\mathfrak{P}_{\epsilon}^{(1)}g_s,\mathfrak{P}_{\epsilon}^{(1)}g_s\Big)dz \Big|\\
&\le C \|\mathbb{K}^{(21)}\|^2\|g_s\|^2\sup_z\|(\mathbb{K}^{(22)}-z)^{-1}\|\cdot \oint_{\mathcal{L}}  \|\mathbb{G}^{(11)}(z)\|\cdot \|(\mathbb{K}^{(11)}-z)^{-1}\||dz| 
\\
&\le C(\epsilon^2/W^3)\cdot W\cdot N=C/W=o(1),\\
&\Big|\oint_{\mathcal{L}} z^{N-1}
\Big(\big(\mathbb{G}^{(22)}(z)-(\mathbb{K}^{(22)}-z)^{-1}\big)\mathfrak{P}_{\epsilon}^{(2)}g_s,\mathfrak{P}_{\epsilon}^{(2)}g_s\Big)dz \Big|\\
&\le C \|\mathbb{K}^{(21)}\|^2\|g_s\|^2\cdot \sup_z\big(\|(\mathbb{K}^{(22)}-z)^{-1}\| \|\mathbb{G}^{(22)}(z)\|\big) \oint_{\mathcal{L}} \|(\mathbb{K}^{(11)}-z)^{-1}\||dz| \\
&\le C\varepsilon^2/W^3\cdot W^2\cdot \log N=C\log N/N=o(1).
\end{align*}
Hence,
\begin{align*}
(\mathbb{K}^{N-1}g_s,g_s)=&
-\frac{1}{2\pi i}\oint_{\mathcal{L}}z^{N-1}((\mathbb{K}^{(11)}-z)^{-1}\mathfrak{P}_{\epsilon}^{(1)}g_s,\mathfrak{P}_{\epsilon}^{(1)}g_s)dz\\
&-\frac{1}{2\pi i}\oint_{\mathcal{L}}z^{N-1}((\mathbb{K}^{(22)}-z)^{-1}\mathfrak{P}_{\epsilon}^{(2)}g_s,\mathfrak{P}_{\epsilon}^{(2)}g_s)dz
+o(\lambda_{\max}^{N-1}).
\end{align*}
Observe that the second integral here is zero, since $(\mathbb{K}^{(22)}-z)^{-1}$ is analytic in $\Omega$.
Then, by the Cauchy theorem, we obtain (\ref{l:P.1}) with  $g_0=\mathcal{P}_0 g$ replaced by 
$\mathfrak{P}_{\epsilon}^{(1)}g_s$.

But it is easy to see that
\[ 
\|\mathfrak{P}_{\epsilon}^{(1)}g_s-g_0\|  = \|\mathfrak{P}_{\epsilon}^{(1)}g_s-\mathcal{P}_0g_s\|=O(W^{-1/2}).\]
This completes the proof of the lemma.

$\square$

\textit{Poof of Theorem \ref{t:new}}. 
Consider the matrix $\mathbb{A}^{(M)}$ of the form
\begin{align}\label{A^M}
\mathbb{A}^{(M)}_{\ell,k;\ell',k'}=-\delta_{\ell,\ell'}\delta_{k,k'} \ell(\ell+1)/8(u_*\kappa_*)^2+(\nu t^{(\ell)}_{0k}, t^{(\ell')}_{0k'}).
\end{align}
By (\ref{cor:1.2})   we have 
\[
|\mathbb{K}^{(11)}-I-N^{-1}\mathbb{A}^{(M)}\|\le C(\log W)^4/W^{5/2}.
\]
Hence,
\[
((\mathbb{K}^{(11)})^{N-1}g_0,g_0)=((I+N^{-1}\mathbb{A}^{(M)})^{N-1}g_0,g_0)+o(1)=(e^{\mathbb{A}}g_0,g_0),
\]
where $\mathbb{A}$ is defined by (\ref{A^M}) without restrictions $\ell,\ell'\le M$. Moreover, since
\[
(\nu t^{(\ell)}_{0k}, t^{(\ell')}_{0k'})=\delta_{k,k'}(\nu t^{(\ell)}_{0k}, t^{(\ell')}_{0k})
\]
and $g_0\sim t^{(0)}_{00}$, we obtain that the space $\mathcal{E}_{00}=\{t^{\ell}_{00}\}_{\ell=0}^\infty$ reduces $\mathbb{A}$.  If 
we take the basis $\{P_{\ell}(z)\}_{\ell=1}^\infty$ in $\mathcal{E}_{00}$, where $P_{\ell}$ is the $\ell$th Legendre polynomial and $z=\cos\theta$,
we obtain that $\mathbb{A}\Big|_{\mathcal{E}_{00}}=\mathbb{A}_0$ with $\mathbb{A}_0$ from (\ref{bbA}).
This completes the proof of Theorem \ref{t:new}.

$\square$


\begin{thebibliography}{99}

\bibitem{Af:16} Afanasiev, I.: On the correlation functions of the characteristic polynomials of the sparse hermitian random matrices.
J. Stat. Phys, \textbf{163},   324 -- 356  (2016)

\bibitem{Af:19} Afanasiev, I.: On the correlation functions of the characteristic
polynomials of non-Hermitian random matrices with
independent entries, J.Stat.Phys. 176, 1561 - 1582 (2019)



\bibitem{AS:sp_det}  Afanasiev, I. ,  Shcherbina, T.: Characteristic polynomials of sparse non-Hermitian random matrices, 	J Stat Phys  \textbf{192:12} (2025)
 
 \bibitem{APSo:09}
Akemann, G. ,Phillips, M.J. and Sommers, H.-J.: Characteristic polynomials in real Ginibre ensembles, J. Phys. A: Math. Theor. 42 (2009)
 
\bibitem {Ak-Ve:03} Akemann, G. , Vernizzi, G., Characteristic Polynomials of Complex Random Matrix Models, Nucl. Phys. B, 3:600, p. 532-556 (2003)


\bibitem{Br-Hi:00} Br\'ezin, E., Hikami, S.: Characteristic polynomials of random matrices.
Commun. Math. Phys. 214, p. 111 -- 135 (2000)

\bibitem{Br-Hi:01} Br\'ezin, E., Hikami, S.: Characteristic polynomials of real symmetric random matrices.
Commun. Math. Phys., vol. 223, p. 363 -- 382 (2001)

%

\bibitem{Dr:band} Drogin, R. Localization of one-dimensional random band matrices, \url{ arXiv:2508.05802v2 } (2025)

\bibitem{YY:2d_band}
Dubova, S., Yang, K., Yau, H.-T., and Yin, J.: Delocalization of Two-Dimensional Random Band Matrices, arXiv:2503.07606

\bibitem{YY:3d_band}
Dubova, S., Yang, F., Yau, H.-T., and Yin, J.: Delocalization of Non-Mean-Field Random Matrices in Dimensions $d\ge 3$, 	arXiv:2507.20274


\bibitem{ErR:band} Erd\H{o}s, L., Riabov, V. The Zigzag Strategy for Random Band Matrices
\url{arXiv: 2506.06441} 


\bibitem{FM:91}
Fyodorov, Y.V., Mirlin, A.D.: Scaling properties of localization in random band matrices: a $\sigma$-model
approach, Phys. Rev. Lett. {\bf 67}, 2405 -- 2409 (1991)

\bibitem{H:24} Han, Y.: The circular law for random band matrices: improved bandwidth for general models, 	arXiv:2410.16457

\bibitem{H:25} Han, Y.:  The circular law for non-Hermitian random band matrices up to bandwidth $W^{1/2+c}$, \url{arXiv:2508.18143}

\bibitem{H:25_1} Han, Y.: Circular law for non-Hermitian block band matrices with slowly growing bandwidth, arXiv:2511.01744 

\bibitem{Hu:63} Hua, L. K. Harmonic Analysis of Functions of Several Complex Variables in the Classical Domains, American Mathematical Society, Providence, RI, 1963

\bibitem{JJLO:21} 
Jain, V., Jana, I., Luh, K., and O'Rourke, S.: Circular law for random block band matrices with genuinely
sublinear bandwidth, J. Math. Phys. \textbf{62:8} (2021)

 \bibitem{MO:24}
Maltsev, A., Osman, M. Bulk universality for complex non-hermitian matrices with independent
and identically distributed entries, \url{ arXiv: 2310.11429v4} (2023)

%

 

%
%



%


\bibitem{SS:17}  Shcherbina, M., Shcherbina, T.: Characteristic polynomials for 1d random band matrices from the localization side,
Commun. Math. Phys. {\bf 351}, 
p. 1009 -- 1044 (2017)


\bibitem{SS:Un} Shcherbina, M.,  Shcherbina, T.:  Universality for 1 d random band matrices, 	Commun. Math. Phys.{\bf 385}, 667 -- 716 (2021)

\bibitem{SS:band_com_d}  Shcherbina, M., Shcherbina, T.: Characteristic polynomials of non-Hermitian random band
matrices,\url{arXiv:2510.04255}


\bibitem{TSh:ChW}
Shcherbina, T.: On the correlation function of the characteristic polynomials of the Hermitian Wigner
ensemble. Commun. Math. Phys., vol. 308, p. 1 -- 21 (2011),

\bibitem{TSh:ChSC}
 Shcherbina, T.: On the correlation functions of the characteristic polynomials of the hermitian sample
covariance ensemble, Probab. Theory Relat. Fields, vol. 156, p. 449 -- 482 (2013)

\bibitem{TSh:14}
 Shcherbina, T. : On the second mixed moment of the characteristic polynomials of the 1D band matrices.
 Commun. Math. Phys., vol.  328, p. 45 -- 82 (2014)
 
 \bibitem {TSh:15} Shcherbina, T.. Universality of the second mixed moment of the characteristic polynomials of the 1D band matrices:
real symmetric case, J. Math. Phys.   56, pp. 29 (2015)
 
 \bibitem{TS:20} Shcherbina, T.: Characteristic polynomials of random band matrices near the threshold, J.Stat.Phys. 179:4, p. 920 -- 944 (2020)


\bibitem{TS:22} Shcherbina, T.: Transfer matrix approach for the real symmetric 1D random band matrices,
Electron. J. Probab.  27, p. 1-29 (2022)



 \bibitem {St-Fy:03} Strahov and Y. V. Fyodorov, Y.V., Strahov, E. Universal results for correlations of  characteristic polynomials:  Riemann-Hilbert approach
	Comm. Math. Phys., 241:2-3, p. 343-382 (2003)

\bibitem{TV:08}
Tao,T., Vu, V.: Random matrices: the circular law. Commun. Contemp. Math., 10(2):261–307 (2008)

\bibitem{Ta-Vu:15}
Tao, T., Vu, V.: Random matrices: universality of local spectral statistics of non-
Hermitian matrices, Ann. Probab. 43, 782–874 (2015)

\bibitem{TVKr:10} Tao T., Vu V., Krishnapur M., Random matrices: Universality of ESDs and the circular law, Annals of Probability 38:5,
2023  (2010)

\bibitem{Tikh:23}
Tikhomirov, K.: On pseudospectrum of inhomogeneous non-Hermitian random
matrices, arXiv:2307.08211 (2023)

\bibitem{Vil:68}  Vilenkin, N. Ja.: Special Functions and the Theory of Group Representations.
Translations of Mathematical Monographs, AMS
1968; 613 pp; 





\bibitem{YY:1d_band}
Yau, H.-T., Yin, J.: Delocalization of one-dimensional random band matrices, arXiv:2501.01718

\end{thebibliography}
\end{document}